\documentclass[journal]{IEEEtran}
\usepackage[T1]{fontenc}

\usepackage{amssymb}
\usepackage{amsmath}
\usepackage{graphicx}
\usepackage{cite}
\usepackage{url}
\usepackage{fancyhdr}
\usepackage{caption}
\usepackage{amsthm}
\usepackage{setspace}
\usepackage{hyperref}
\usepackage{algorithm}
\usepackage{algorithmic}
\usepackage{multirow}
\usepackage{makecell}
\usepackage{mathtools}
\usepackage{subcaption}
\usepackage{bm}
\usepackage{tikz}
\usepackage{xurl}
\usepackage{stfloats}
\usepackage{mathrsfs}

\usepackage{booktabs}
\usepackage{siunitx}
\usepackage{balance}
\usepackage{array}
\usepackage{tabularx}

\hypersetup{colorlinks=true,
	linkcolor=blue,
	citecolor=blue,      
	urlcolor=black,
}

\expandafter\def\expandafter\normalsize\expandafter{%
	\normalsize%
	\setlength\abovedisplayskip{5.5pt}%
	\setlength\belowdisplayskip{5.5pt}%
	\setlength\abovedisplayshortskip{3.5pt}%
	\setlength\belowdisplayshortskip{3.5pt}%
}

\allowdisplaybreaks

\theoremstyle{definition}

\newtheorem{theorem}{Theorem}
\newtheorem{lemma}{Lemma}
\newtheorem{remark}{Remark}
\newtheorem{corollary}{Corollary}

\graphicspath{{./figure/}}

\begin{document}

\title{Modeling and Analysis for Joint Design of Communication and Control}

	\author{Xu Gan,~\IEEEmembership{Member, IEEE}, Chongjun Ouyang,~\IEEEmembership{Member, IEEE}, and Yuanwei Liu,~\IEEEmembership{Fellow,~IEEE}

	\vspace{-2mm}

\thanks{Xu Gan and Yuanwei Liu are with the Department of Electrical and Computer Engineering, The University of Hong Kong, Hong Kong (e-mail: \{eee.ganxu, yuanwei\}@hku.hk).}
\thanks{Chongjun Ouyang is with the School of Electronic Engineering and Computer Science, Queen Mary University of London, London E1 4NS, U.K. (e-mail: c.ouyang@qmul.ac.uk).}
}

\maketitle

\begin{abstract}
A unified analytical framework for joint design of communication and control (JDCC) is proposed. Within this framework, communication transmission delay and steady-state control variance are derived as the two fundamental JDCC performance metrics. The Pareto boundary is then established to characterize the optimal communication-control trade-off in JDCC systems. To further obtain closed-form expressions, their performance regions are derived under maximum-ratio transmission (MRT) and zero-forcing (ZF) beamforming. For system reliability evaluation, the communication-only and control-only outage probabilities are first derived. Based on these, the JDCC outage probability is defined to quantify the probability that the communication-delay and control-error requirements cannot be simultaneously satisfied. Its analytical expressions are then derived under both MRT and ZF schemes. Finally, numerical results validate the theoretical results and reveal that: (1) the Pareto boundary characterizes the trade-off frontier and performance limit of JDCC systems and (2) the JDCC reliability is jointly determined by the uplink-downlink closed-loop control and its coupling with communication.
\end{abstract}

\begin{IEEEkeywords}
Joint design of communication and control (JDCC), outage probability, Pareto boundary, rate-distortion theory.
\end{IEEEkeywords}

\section{Introduction}
\label{sec:intro}

\IEEEPARstart{T}{he} proliferation of emerging applications, such as industrial automation, robotic manipulation, and unmanned aerial vehicles (UAVs), has imposed increasingly stringent requirements on wireless systems\cite{aceto2019survey,tataria20216g,zeydan20246g}. Different from conventional data-centric services, these tasks require real-time closed-loop interaction and control with physical devices. As such, the base station (BS) is expected to be equipped with control capability\cite{park2017wireless}, thereby leveraging the existing large-scale communication infrastructure for these emerging applications. Under this paradigm, the wireless channel is no longer merely a medium for data transmission, but becomes an integral part of the control loop itself. This trend is driving next-generation wireless networks beyond communication-only architectures toward integrated platforms that can jointly support efficient data services and reliable remote control\cite{wang2023review}.

To meet such requirements, joint design of communication and control (JDCC)\cite{liu2017joint,jahangiri2025integrated,10891172} has emerged as a promising paradigm, where communication and control are jointly designed within a unified wireless system and share the same hardware platform and spectrum resources. However, such integration also induces fundamental coupling and trade-off between communication and control. This calls for a well-designed JDCC framework that enables the wireless infrastructure to simultaneously support low-latency communication and low-error control. Despite its importance, such a unified modeling and analytical framework is still absent in the existing literature.

\subsection{Prior Works}
To better clarify the distinction between the present work and the existing literature, related studies are briefly reviewed from the following three perspectives: 1) communication-centric works, 2) control-centric works, and 3) JDCC-related works.

\subsubsection{Communication-Centric Works}
The first category includes works that support control-related applications through enhanced wireless transmission design, while still optimizing communication performance as the main objective. For instance, the authors of \cite{voigtlander20175g} examined 5G ultra-reliable low-latency communication (URLLC) for distributed robotic systems, with the main focus on whether the wireless link can provide sufficiently low delay and high reliability for remote operation. In \cite{wu2023latency}, UAV-enabled mobile edge computing was designed for mission-critical tasks by minimizing the maximum computation latency through URLLC-based offloading and joint resource optimization. In heterogeneous robotic systems, the authors of \cite{lv2023multi} developed a distributed 5G architecture to improve multi-robot cooperation via efficient wireless information exchange. In addition, survey and perspective works such as \cite{liu2021robotic}, \cite{mu2021intelligent}, and \cite{masaracchia2021uav} have emphasized the role of reliable and low-latency wireless connectivity in robotic and UAV applications. However, these works do not establish a theoretical framework that explicitly characterizes the mathematical relationship between control performance and wireless network parameters.

\subsubsection{Control-Centric Works}
The second category includes works that analyze control performance under simplified communication processes. In this direction, the authors of \cite{ulusoy2010wireless} studied wireless model-based predictive control performance by jointly incorporating packet deadlines, actuator buffering, and predictive control to tolerate packet losses and variable delays. In \cite{walsh2002stability}, the communication effect was abstracted into scheduling-induced error bounds and transfer intervals for control stability analysis. The experimental study in \cite{ploplys2004closed} examined closed-loop control over IEEE 802.11b networks, where the wireless effect was mainly represented by delay, packet loss, and adaptive sampling. In \cite{pajic2011wireless}, the wireless control network paradigm was developed and mean-square stability was analyzed under packet drops and link failures. However, these works do not account for physical-layer channels or transceiver beamforming, and hence cannot explicitly reveal the control performance under realistic wireless networks.

\subsubsection{JDCC-Related Works}
The third category includes works that study the co-design of communication and control. For instance, the authors of \cite{casanova2003multirate} redesigned the control architecture to compensate communication-induced delay and timing uncertainty over a shared industrial medium. In~\cite{gaid2006optimal}, control and scheduling were jointly optimized over a bandwidth-limited network to enhance control quality under communication constraints. In a wireless vehicular scenario, the authors of \cite{zeng2018integrated} incorporated transmission delay and interference into stability analysis, and then mapped the resulting control requirements into wireless communication constraints. More recently, \cite{zhou2025integrated} jointly designed sensing, communication, and control for multi-AGV closed-loop systems in industrial environments. However, these studies still focus on the joint-design methods mainly by imposing communication constraints on control or control constraints on communication, rather than by explicitly capturing their mutual interaction and the resulting trade-off performance.

\subsection{Motivations and Contributions}
Most existing works examine communication support, control performance, or communication-control co-design in wireless systems, but still lack a unified analytical framework for explicitly characterizing their intrinsic coupling. As a result, two fundamental questions remain open: \emph{how to characterize the communication-control trade-off in JDCC systems}, and \emph{how to quantify the reliability of their joint operation}. Motivated by these questions, this paper develops a unified modeling and analytical framework for JDCC systems. In particular, we characterize the Pareto boundary of the JDCC performance region and establish the joint outage probability. The main contributions of this paper are summarized as follows.
\begin{itemize}
  \item We develop a unified analytical framework for JDCC systems, where a multi-antenna BS simultaneously serves a communication user (CU) and a controllable device (CD) in the same wireless network, as illustrated in Fig.~\ref{fig:illustration}. In particular, the control functionality is explicitly modeled as an uplink-downlink closed loop through state reporting and control-command delivery, while the communication functionality is served in parallel over shared wireless resources.

  \item Within the developed framework, we derive communication transmission delay and steady-state control variance as the two fundamental JDCC performance metrics. By leveraging rate-distortion theory, we establish the mathematical relationship between the uplink/downlink transmission capacities and the control variance, from which the state-variance evolution equation is derived. This equation is then used to derive the closed-form steady-state control variance and the corresponding stability conditions. We then analyze its asymptotic behavior to identify the dominant wireless-link bottleneck under different operating regimes.

  \item Using the derived communication and control metrics, we characterize the Pareto boundary of JDCC performance, which reveals the optimal trade-off under shared wireless resources. We then derive closed-form analytical expressions for the communication-control performance regions under two typical beamforming schemes, namely, maximum-ratio transmission (MRT) and zero-forcing (ZF). Based on these analytical results, we further compare the Pareto boundary with the MRT and ZF performance regions under different transmit-power regimes and channel correlations.

  \item We define and derive the communication-only and control-only outage probabilities as reliability metrics for single-function JDCC systems. Based on these results, we further establish the joint outage probability as a new reliability metric for JDCC systems, which is defined as the probability that the JDCC system cannot simultaneously satisfy the communication-delay and control-error requirements. We then derive the analytical expressions of the JDCC outage probability under MRT and ZF schemes, thereby revealing the reliability impact of communication-control coupling.
  
\end{itemize}

\subsection{Organization and Notations}
The remainder of this paper is organized as follows. Section~\ref{sec:system} presents the JDCC system model and develops the communication and control performance metrics. Section~\ref{sec:beamforming} characterizes the achievable communication-control Pareto-optimal boundary, together with closed-form performance region results under MRT and ZF beamforming. Section~\ref{sec:outage} investigates the outage behavior of the JDCC system and derives the single-function and joint outage probabilities. Section~\ref{sec:numerical} provides numerical results to validate the theoretical results and illustrate the main design insights. Finally, Section~\ref{sec:conclusion} concludes the paper.

\emph{Notation}: Throughout this paper, italic letters denote scalars, whereas boldface lowercase and uppercase letters denote vectors and matrices, respectively. The set of complex-valued $M \times N$ matrices is denoted by $\mathbb{C}^{M \times N}$. For any vector or matrix, $(\cdot)^T$ and $(\cdot)^H$ represent the transpose and Hermitian transpose, respectively. The symbols $\mathbf{I}_M$ and $\mathbf{0}$ denote the $M \times M$ identity matrix and the all-zero vector or matrix of appropriate dimension. Moreover, $[\mathbf{h}]_m$ denotes the $m$-th entry of $\mathbf{h}$, $\|\cdot\|$ denotes the Euclidean norm, and $|\cdot|$ denotes the magnitude of a scalar. The operators $\mathbb{E}[\cdot]$ and $\Pr(\cdot)$ stand for expectation and probability, respectively, $\triangleq$ denotes equality by definition, and $[x]^+ \triangleq \max(x,0)$. Finally, $\mathcal{CN}(0,\sigma^2)$ and $\mathcal{CN}(\mathbf{0},\sigma^2\mathbf{I}_M)$ denote scalar and vector circularly symmetric complex Gaussian distributions, respectively.

\section{System Model and Performance Metrics}
\label{sec:system}
This section introduces the system model and performance metrics of the proposed JDCC system. We first present the system architecture and the associated uplink-downlink transmission protocol, under which the BS simultaneously supports the communication functionality for the CU and the closed-loop control functionality for the CD. We then characterize the corresponding signal transmission models and define the communication transmission delay and the steady-state control variance as the fundamental performance metrics of the two functionalities.

\begin{figure}[t]
  \centering
  \includegraphics[width=\linewidth]{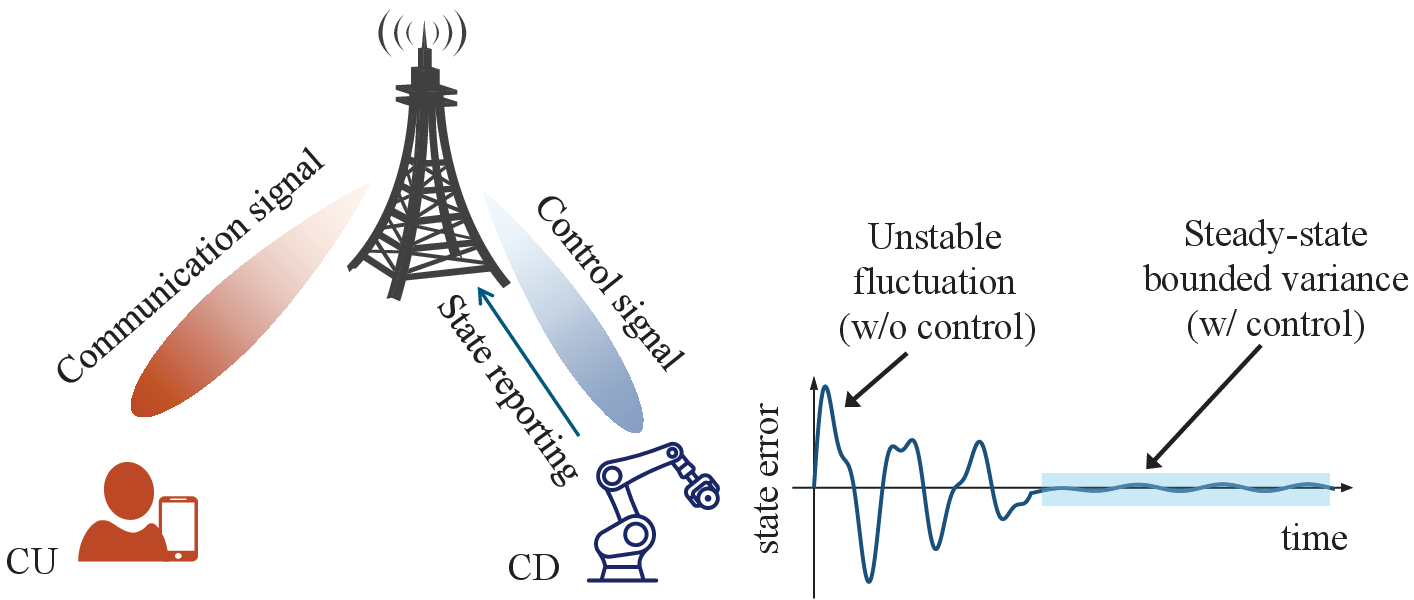}
  \caption{Illustration of the proposed JDCC system.}
  \label{fig:illustration}
\end{figure}

\subsection{JDCC System Model}
We consider the JDCC system as illustrated in Fig.~\ref{fig:illustration}, where an $M$-antenna BS simultaneously serves one CU and one CD, with $M \ge 2$. The BS delivers the communication signal to the CU, while supporting the control function of the CD through an uplink-downlink closed loop. Specifically, the CD first reports its state to the BS via the uplink, and the BS then sends the corresponding control signal back to the CD through the downlink. As shown in Fig.~\ref{fig:illustration}, the control objective is to suppress the unstable state fluctuation of the CD and maintain a bounded steady-state variance. To capture the dynamics of the controlled process, the CD is modeled as a scalar, discrete-time linear time-invariant (LTI) system, whose state evolves as
\begin{equation}\label{eq:state_equation}
  x_{n+1} = a\,x_n + b\,u_n + w_n,
\end{equation}
where $x_n,\, u_n,\, w_n \in \mathbb{C}$ are the process state, the input control signal, and the process noise at the $n$-th control interval, respectively. The evolving parameters $a, \, b \in \mathbb{C}$ are assumed constant across time. The process noise $w_n \sim \mathcal{CN}(0,\sigma_w^2)$ is i.i.d.\ across time and independent of all channel noises.
\begin{remark}
  The state transition equation \eqref{eq:state_equation} can be used to model a mobile robot arm or a UAV executing a trajectory-tracking task. In this case, the process state $x_n = x_{n,\mathrm{R}} + j x_{n,\mathrm{I}}$ represents the 2D positional tracking error at the $n$-th control interval. We consider the case $|a| > 1$ and $b \neq 0$. In this case, without the corrective input $u_n$ transmitted by the BS, the state error diverges exponentially and becomes unbounded over time. This makes the state variance:
  \begin{equation}
    V_n \triangleq \mathbb{E}[|x_n|^2],
  \end{equation}
  naturally serve as the control performance metric for quantifying the mean-square state fluctuation.
\end{remark}

\subsection{JDCC Transmission Protocol}
To support both downlink communication for the CU and uplink-downlink closed-loop control for the CD, the JDCC system operates over two timescales: the control interval indexed by $n$ with sampling period $T_s^D$, and the communication update period $T_s^U$. Each control interval consists of an uplink state-reporting phase and a downlink dual-functional transmission phase. In the uplink phase, bandwidth $B_{\mathrm{up}}$ is allocated for the CD to report its state to the BS. In the downlink phase, bandwidth $B_{\mathrm{dn}}$ is used by the BS to simultaneously deliver the control command to the CD and the communication data to the CU. Both uplink state reporting and downlink control delivery are modeled in a block-based manner. Specifically, within the $n$-th control interval of duration $T_s^D$, the current state sample and the corresponding control command are conveyed through coded uplink and downlink transmission blocks, respectively.

\subsubsection{Uplink CD State Reporting}

In the uplink phase, the CD reports the current state sample $x_n$ to the BS through an uplink transmission block associated with the $n$-th control interval. Let $\mathbf{s}_{D,n}^{\mathrm{up}} \in \mathbb{C}^{L_{\mathrm{up}} \times 1}$ denote the coded uplink signal block, where $L_{\mathrm{up}}$ is the uplink block length within one control interval. The received uplink signal block at the BS is
\begin{equation}
  \mathbf{Y}_{n}^{\mathrm{up}} = \mathbf{h}_D (\mathbf{s}_{D,n}^{\mathrm{up}})^T + \mathbf{Z}_{n}^{\mathrm{up}},
  \label{eq:uplink_io}
\end{equation}
where $\mathbf{h}_D \in \mathbb{C}^{M \times 1}$ is the uplink channel vector, and $\mathbf{Z}_{n}^{\mathrm{up}} \in \mathbb{C}^{M \times L_{\mathrm{up}}}$ is the AWGN matrix with independent entries distributed as $\mathcal{CN}(0,\sigma_{\mathrm{up}}^2)$, where $\sigma_{\mathrm{up}}^2 = N_0 B_{\mathrm{up}}$. The uplink block satisfies the average power constraint $\mathbb{E}[\|\mathbf{s}_{D,n}^{\mathrm{up}}\|^2] \le L_{\mathrm{up}} P_{\mathrm{up}}$.

The BS applies maximum-ratio combining (MRC) to the received uplink block. The corresponding uplink SNR is
\begin{equation}
  \mathrm{SNR}_D^{\mathrm{up}} = \frac{P_{\mathrm{up}}\,\|\mathbf{h}_D\|^2}{\sigma_{\mathrm{up}}^2}.
  \label{eq:snr_up}
\end{equation}
This SNR determines the reliable information rate available for state reporting within the $n$-th control interval. Accordingly, the state available at the BS is interpreted as a rate-limited reconstruction of $x_n$.

\subsubsection{Downlink Dual-Functional Transmission}
\label{subsubsec:downlink_two_timescale}
After receiving the uplink state description, the BS computes the corresponding control command for the CD while simultaneously transmitting communication data to the CU. Accordingly, the downlink signal block in the $n$-th control interval is
\begin{equation}
  \mathbf{S}_{n}^{\mathrm{dn}} = \mathbf{w}_D (\mathbf{s}_{D,n}^{\mathrm{dn}})^T + \mathbf{w}_U (\mathbf{s}_{U,\kappa(n)}^{\mathrm{dn}})^T,
  \label{eq:tx_signal}
\end{equation}
where $\mathbf{s}_{D,n}^{\mathrm{dn}},\, \mathbf{s}_{U,\kappa(n)}^{\mathrm{dn}} \in \mathbb{C}^{L_{\mathrm{dn}} \times 1}$ denote the coded signal blocks for control delivery and communication transmission, respectively, with normalized block powers satisfying $\mathbb{E}[\|\mathbf{s}_{D,n}^{\mathrm{dn}}\|^2] = \mathbb{E}[\|\mathbf{s}_{U,\kappa(n)}^{\mathrm{dn}}\|^2] = L_{\mathrm{dn}}$. Here, $L_{\mathrm{dn}}$ is the downlink block length within one control interval, and $\mathbf{w}_D, \mathbf{w}_U \in \mathbb{C}^{M \times 1}$ are the beamforming vectors for the CD and the CU, respectively, subject to the sum-power constraint $\|\mathbf{w}_D\|^2 + \|\mathbf{w}_U\|^2 \le P_{\mathrm{dn}}$. The communication block paired with the $n$-th control update is indexed by $\kappa(n)$, whose explicit form is not needed in the sequel.

\paragraph{Downlink Transmission for the CU}
Let $\mathbf{h}_U \in \mathbb{C}^{M\times 1}$ denote the downlink channel vector from the BS to the CU. The received signal block at the CU is
\begin{equation}
  \mathbf{y}_{U,n}^{\mathrm{dn}} = \mathbf{h}_U^H \mathbf{w}_U\, \mathbf{s}_{U,\kappa(n)}^{\mathrm{dn}} + \mathbf{h}_U^H \mathbf{w}_D\, \mathbf{s}_{D,n}^{\mathrm{dn}} + \mathbf{z}_{U,n}^{\mathrm{dn}},
  \label{eq:y_U}
\end{equation}
where $\mathbf{z}_{U,n}^{\mathrm{dn}} \in \mathbb{C}^{1 \times L_{\mathrm{dn}}}$ is the AWGN vector with independent entries distributed as $\mathcal{CN}(0,\sigma_{\mathrm{dn}}^2)$, and $\sigma_{\mathrm{dn}}^2 = N_0 B_{\mathrm{dn}}$. By treating the control signal as interference, the downlink SINR for the CU is
\begin{equation}
  \mathrm{SINR}_U^{\mathrm{dn}} = \frac{|\mathbf{h}_U^H \mathbf{w}_U|^2}{|\mathbf{h}_U^H \mathbf{w}_D|^2 + \sigma_{\mathrm{dn}}^2}.
  \label{eq:sinr_U}
\end{equation}
This SINR determines the communication rate available to the CU.

\paragraph{Downlink Transmission for the CD}
Under time-division duplex (TDD) reciprocity, the BS-to-CD downlink channel is represented by $\mathbf{h}_D^H$. The received signal block at the CD is
\begin{equation}
  \mathbf{y}_{D,n}^{\mathrm{dn}} = \mathbf{h}_D^H \mathbf{w}_D\, \mathbf{s}_{D,n}^{\mathrm{dn}} + \mathbf{h}_D^H \mathbf{w}_U\, \mathbf{s}_{U,\kappa(n)}^{\mathrm{dn}} + \mathbf{z}_{D,n}^{\mathrm{dn}},
  \label{eq:y_D}
\end{equation}
where $\mathbf{z}_{D,n}^{\mathrm{dn}} \in \mathbb{C}^{1 \times L_{\mathrm{dn}}}$ denotes the AWGN vector with independent entries distributed as $\mathcal{CN}(0,\sigma_{\mathrm{dn}}^2)$. By treating the communication signal as interference, the downlink SINR for the CD is
\begin{equation}
  \mathrm{SINR}_D^{\mathrm{dn}} = \frac{|\mathbf{h}_D^H \mathbf{w}_D|^2}{|\mathbf{h}_D^H \mathbf{w}_U|^2 + \sigma_{\mathrm{dn}}^2}.
  \label{eq:sinr_D}
\end{equation}
This SINR determines the reliable rate available for control delivery within the control interval. Accordingly, the control command available at the CD is interpreted as a rate-limited reconstruction of the intended command.

The above protocol specifies the uplink and downlink block-transmission models for the communication and control functions in JDCC systems. Based on the resulting SNR/SINR expressions, we next define the corresponding communication and control performance metrics through an information-theoretic distortion-based characterization.

\subsection{JDCC Performance Metric}
\label{subsec:icac_metric}

The JDCC system is characterized by two metrics: the communication transmission delay and the steady-state control variance.

\subsubsection{Communication Performance Metric}
For the communication function, define the downlink communication SINR as $\Gamma_U \triangleq \mathrm{SINR}_U^{\mathrm{dn}}$. The communication transmission delay for delivering a payload of $Q_U$ bits is
\begin{equation}
	\tau_U  = \frac{Q_U}{B_{\mathrm{dn}}\log_2\!\bigl(1 + \Gamma_U\bigr)}.
	\label{eq:cu_delay}
\end{equation}

\subsubsection{Control Performance Metric}
For the control function, we adopt the steady-state control variance, which quantifies the long-term mean-square fluctuation of the controlled process. The key step is to establish the state-variance recursion under uplink state-reporting distortion and downlink control-delivery distortion. To this end, we first relate the uplink and downlink reliable transmission rates to the corresponding distortions through Gaussian rate-distortion theory under the adopted block-based transmission protocol. We then incorporate these distortions into the closed-loop dynamics to derive the state-variance evolution, from which the steady-state variance, stability condition, and asymptotic regimes as follows.

\paragraph{Uplink State Reconstruction Distortion}
We first characterize the state-reconstruction distortion at the BS. In each control interval, the process state $x_n$ is encoded at the CD and conveyed to the BS through the uplink block. Since only a finite number of bits can be reliably delivered within one control interval, the reconstructed state $\hat{x}_n$ is subject to nonzero distortion. To quantify this distortion, we relate the uplink channel capacity to the Gaussian rate-distortion function. For a zero-mean circularly symmetric complex Gaussian source with variance $\sigma^2$, the rate-distortion function is~\cite{cover2006elements}
\begin{equation}
	R(D) = \log_2\!\left(\frac{\sigma^2}{D}\right), \qquad D \le \sigma^2,
	\label{eq:distortion}
\end{equation}
where $R(D)$ denotes the minimum number of bits per source sample required to achieve mean-square distortion $D$. Under the adopted ideal Gaussian source-channel coding abstraction, the reliable transmission rate within each control interval is determined by the corresponding physical-layer SNR/SINR, and the resulting distortion is given by the Gaussian rate-distortion limit.

Since $x_n \sim \mathcal{CN}(0, V_n)$, the rate-distortion function in~\eqref{eq:distortion} applies directly. Let $D_n^{\mathrm{up}} \triangleq \mathbb{E}[|x_n-\hat{x}_n|^2]$ denote the uplink reconstruction distortion. Then,
\[
R(D_n^{\mathrm{up}})=\log_2\!\left(\frac{V_n}{D_n^{\mathrm{up}}}\right).
\]
On the other hand, the uplink achievable rate is $C_D^{\mathrm{up}} = B_{\mathrm{up}}\log_2\!\bigl(1+\mathrm{SNR}_D^{\mathrm{up}}\bigr)$. Under the adopted block-coding abstraction, the maximum number of reliably conveyed bits over one control interval is $C_D^{\mathrm{up}}T_s^D$. By matching this rate with the Gaussian rate-distortion function and defining $\alpha_{\mathrm{up}} \triangleq B_{\mathrm{up}}T_s^D$, we have $\log_2\!\left(\frac{V_n}{D_n^{\mathrm{up}}}\right)
=
\alpha_{\mathrm{up}}\log_2\!\bigl(1+\mathrm{SNR}_D^{\mathrm{up}}\bigr)$, which yields
\begin{equation}
	D_n^{\mathrm{up}}
	=
	\frac{V_n}{(1+\mathrm{SNR}_D^{\mathrm{up}})^{\alpha_{\mathrm{up}}}}.
	\label{eq:Dn_up}
\end{equation}
Under the Gaussian rate-distortion-optimal reconstruction model, the reconstruction error $\Delta x_n = x_n-\hat{x}_n$ is independent of $\hat{x}_n$, and $\Delta x_n \sim \mathcal{CN}(0,D_n^{\mathrm{up}})$.

\paragraph{Downlink Control Command Distortion}
Based on the reconstructed state $\hat{x}_n$ available at the BS, the control command is generated according to the one-step state-cancellation law~\cite{stengel1994optimal,aastrom2013computer}
\begin{equation}
	d_n = -(a/b)\hat{x}_n.
\end{equation}
Since $\hat{x}_n \sim \mathcal{CN}(0, V_n-D_n^{\mathrm{up}})$, the control command is also complex Gaussian, i.e.,
\[
d_n \sim \mathcal{CN}\!\left(0,\frac{|a|^2}{|b|^2}(V_n-D_n^{\mathrm{up}})\right).
\]
The BS encodes $d_n$ into the downlink block and the CD reconstructs it as $\hat{d}_n$. Define the resulting downlink distortion as $D_n^{\mathrm{dn}} \triangleq \mathbb{E}[|d_n-\hat{d}_n|^2]$. Since $d_n$ is a complex Gaussian source, the same rate-distortion function applies. The achievable downlink rate is $C_D^{\mathrm{dn}} = B_{\mathrm{dn}}\log_2\bigl(1+\mathrm{SINR}_D^{\mathrm{dn}}\bigr)$, and the maximum number of reliably conveyed bits over one control interval is $C_D^{\mathrm{dn}}T_s^D$. By matching this rate with the Gaussian rate-distortion function and defining $\alpha_{\mathrm{dn}} = B_{\mathrm{dn}} T_s^D$, we obtain $\log_2\!\left(\frac{\frac{|a|^2}{|b|^2}(V_n-D_n^{\mathrm{up}})}{D_n^{\mathrm{dn}}}\right)
=
\alpha_{\mathrm{dn}}\log_2\!\bigl(1+\mathrm{SINR}_D^{\mathrm{dn}}\bigr)$, which yields
\begin{equation}
	D_n^{\mathrm{dn}}
	=
	\frac{\frac{|a|^2}{|b|^2}(V_n-D_n^{\mathrm{up}})}{(1+\mathrm{SINR}_D^{\mathrm{dn}})^{\alpha_{\mathrm{dn}}}}.
	\label{eq:Dn_dn}
\end{equation}
Under the Gaussian rate-distortion-optimal reconstruction model, the reconstruction error $\Delta d_n = d_n-\hat{d}_n$ is independent of $\hat{d}_n$. Hence, $\Delta d_n \sim \mathcal{CN}(0,D_n^{\mathrm{dn}})$.

\paragraph{State-Variance Evolution Equation}
We now incorporate the uplink and downlink distortions into the closed-loop dynamics. At the CD, the reconstructed control command $\hat{d}_n$ is applied as the input, i.e., $u_n=\hat{d}_n$. Substituting $u_n = d_n - \Delta d_n = -(a/b)\hat{x}_n - \Delta d_n$ and $x_n = \hat{x}_n + \Delta x_n$ into~\eqref{eq:state_equation}, the next-step state becomes
\begin{equation}
  x_{n+1} = a\,\Delta x_n - b\,\Delta d_n + w_n.
\end{equation}
Taking the variance on both sides, we obtain the state-variance recursion. Under the Gaussian rate-distortion-optimal reconstruction model, $\Delta x_n$ is independent of $\hat{x}_n$. Since $d_n$ is a deterministic function of $\hat{x}_n$, $\Delta x_n$ is independent of $d_n$. Moreover, $\Delta d_n$ is independent of $d_n$, and hence independent of $\Delta x_n$. The process noise $w_n$ is also independent of all communication and control variables. Therefore, all cross terms vanish, and the state variance evolves as
\begin{equation}
  V_{n+1} = |a|^2 D_n^{\mathrm{up}} + |b|^2 D_n^{\mathrm{dn}} + \sigma_w^2.
  \label{eq:variance_recursion}
\end{equation}
This recursion shows that the state variance is governed by the uplink and downlink distortions induced by the two control links. Depending on their qualities, the variance may either diverge or converge to a finite steady-state value. This leads to the corresponding stability condition and steady-state variance characterized next.

\begin{theorem}[Steady-State Control Variance]
  \label{thm:main}
  The state variance sequence $\{V_n\}$ converges to a unique steady-state value
  \begin{equation}
  	V_\infty \triangleq \lim_{n\to\infty} V_n
  	= \frac{\sigma_w^2\,S_\alpha\,\Gamma_\alpha}
  	{S_\alpha\,\Gamma_\alpha - |a|^2(S_\alpha + \Gamma_\alpha - 1)},
  	\label{eq:V_infinity}
  \end{equation}
  if the following stability conditions are satisfied:
  \begin{subequations}
  	\label{eq:stability_condition}
  	\begin{align}
  		& S_\alpha > |a|^2, \label{eq:stab_up}\\
  		& \Gamma_\alpha > \frac{|a|^2(S_\alpha - 1)}{S_\alpha - |a|^2},
  		\label{eq:stab_dn}
  	\end{align}
  \end{subequations}
  where $S_\alpha \triangleq (1+\mathrm{SNR}_D^{\mathrm{up}})^{\alpha_{\mathrm{up}}}$ and $\Gamma_\alpha \triangleq (1+\mathrm{SINR}_D^{\mathrm{dn}})^{\alpha_{\mathrm{dn}}}$.
\end{theorem}

\begin{IEEEproof}
  Substituting $D_n^{\mathrm{up}}$ and $D_n^{\mathrm{dn}}$ into \eqref{eq:variance_recursion}, the variance recursion forms a first-order affine map:
  \begin{equation}
    V_{n+1}
    = 
    |a|^2\!\left(\frac{1}{S_\alpha} + \frac{1}{\Gamma_\alpha} - \frac{1}{S_\alpha\,\Gamma_\alpha}\right)
    V_n + \sigma_w^2.
    \label{eq:Vn_affine}
  \end{equation}
  The closed-loop JDCC system is mean-square stable if and only if $\frac{|a|^2(S_\alpha+\Gamma_\alpha-1)}{S_\alpha\Gamma_\alpha}< 1$, i.e., $S_\alpha\Gamma_\alpha > |a|^2(S_\alpha + \Gamma_\alpha - 1)$. This can be rewritten as $(S_\alpha - |a|^2)(\Gamma_\alpha - |a|^2) > |a|^2(|a|^2 - 1)$, which requires $S_\alpha > |a|^2$ and then $\Gamma_\alpha > |a|^2(S_\alpha - 1)/(S_\alpha - |a|^2)$, yielding~\eqref{eq:stability_condition}. Moreover, setting $V_{n+1} = V_n = V_\infty$ yields the fixed point in~\eqref{eq:V_infinity} as the unique steady-state value.
\end{IEEEproof}

To gain further insight, it is useful to examine the asymptotic behavior of $V_\infty$ under different uplink and downlink quality regimes. This analysis reveals how the control performance behaves when one or both links become sufficiently strong, and identifies which link acts as the dominant bottleneck in each regime.

\begin{corollary}[Asymptotic Regimes of the Steady-State Variance]
	\label{cor:asymptotic}
	The steady-state variance $V_\infty$ in \eqref{eq:V_infinity} exhibits the following asymptotic behaviors.
	
	\begin{enumerate}
		\item \emph{High-SNR regime on both uplink and downlink:}  
		If $S_\alpha \to \infty$ and $\Gamma_\alpha \to \infty$, then
		\begin{equation}
			V_\infty \to \sigma_w^2.
		\end{equation}
		
		\item \emph{High uplink SNR with finite downlink SINR:}  
		If $S_\alpha \to \infty$ while $\Gamma_\alpha$ remains finite and satisfies $\Gamma_\alpha > |a|^2$, then
		\begin{equation}
			V_\infty \to
			\frac{\sigma_w^2\,\Gamma_\alpha}
			{\Gamma_\alpha-|a|^2}.
		\end{equation}
		
		\item \emph{High downlink SINR with finite uplink SNR:}  
		If $\Gamma_\alpha \to \infty$ while $S_\alpha$ remains finite and satisfies $S_\alpha > |a|^2$, then
		\begin{equation}
			V_\infty \to
			\frac{\sigma_w^2\,S_\alpha}
			{S_\alpha-|a|^2}.
		\end{equation}
	\end{enumerate}
\end{corollary}

\section{Pareto Boundary Characterization of JDCC Systems}
\label{sec:beamforming}

After establishing the communication and control performance metrics of the JDCC system, we now characterize the communication-control performance region in the $(\tau_U,V_\infty)$ plane, with particular focus on its Pareto boundary. This section concentrates on the downlink trade-off, while assuming that the uplink control link is sufficiently reliable to guarantee control stability, e.g., $S_\alpha > \frac{|a|^2(\Gamma_\alpha^{\max} - 1)}{\Gamma_\alpha^{\max} - |a|^2}$, where $\Gamma_\alpha^{\max} =
(1+P_{\mathrm{dn}}\|\mathbf{h}_D\|^2/\sigma_{\mathrm{dn}}^2)^{\alpha_{\mathrm{dn}}}$
and $|a|^2 < \Gamma_\alpha^{\max}$. Under this setup, we first identify the single-function performance limits and then characterize the Pareto boundary of JDCC systems.

\subsection{Single-Function Performance Limits}
\label{subsec:separate}

To better understand the communication-control trade-off, it is instructive to first examine the two extreme operating regimes in which the BS serves only one functionality. These two single-function benchmarks correspond to communication-only transmission and control-only transmission, respectively.

\subsubsection{Paradigm A: Communication-Only Performance}

In this regime, the BS allocates all downlink power to the CU and steers the beam toward the CU via MRT, i.e.,

\begin{equation}
  \mathbf{w}_U = \sqrt{P_{\mathrm{dn}}}\,\frac{\mathbf{h}_U}{\|\mathbf{h}_U\|}, \qquad
  \mathbf{w}_D = \mathbf{0}.
  \label{eq:mrt_comm}
\end{equation}
Under this design, the communication transmission delay is minimized and attains
\begin{equation}
  \tau_U^{\mathrm{com}} = \frac{Q_U}{B_{\mathrm{dn}}\log_2\!\left(1 + \frac{P_{\mathrm{dn}}\,\|\mathbf{h}_U\|^2}{\sigma_{\mathrm{dn}}^2}\right)}.
  \label{eq:Du_min}
\end{equation}
At the same time, no control signal is delivered to the CD, i.e.,
$\mathrm{SINR}_D^{\mathrm{dn}} = 0$. As a result, the stability
condition~\eqref{eq:stability_condition} is violated for all $|a|>1$, and the
control process becomes unstable with $V_\infty \to \infty$.

\subsubsection{Paradigm B: Control-Only Performance}

In this regime, the BS allocates all downlink power to the CD and steers the beam toward the CD via MRT, i.e.,
\begin{equation}
  \mathbf{w}_D = \sqrt{P_{\mathrm{dn}}}\,\frac{\mathbf{h}_D}{\|\mathbf{h}_D\|}, \qquad
  \mathbf{w}_U = \mathbf{0}.
  \label{eq:mrt_ctrl}
\end{equation}
Under this design, the control functionality achieves its best possible performance as 
\begin{equation}
	\begin{aligned}
		& V_\infty^{\min} =
		\frac{\sigma_w^2\,S_\alpha\,\Gamma_\alpha^{\max}}
		{S_\alpha\,\Gamma_\alpha^{\max}
			- |a|^2(S_\alpha + \Gamma_\alpha^{\max} - 1)}.
		\label{eq:V_min}
	\end{aligned}
\end{equation}
where $\Gamma_\alpha^{\max} \triangleq
(1+P_{\mathrm{dn}}\|\mathbf{h}_D\|^2/\sigma_{\mathrm{dn}}^2)^{\alpha_{\mathrm{dn}}}$. At the same time, no communication payload can be delivered to the CU, and hence $\tau_U \to \infty$.

\subsection{JDCC Pareto Boundary Characterization}
\label{subsec:pareto}

In practical JDCC systems, communication and control must be supported simultaneously over shared wireless resources. This makes the communication delay $\tau_U$ and the steady-state control variance $V_\infty$ inherently coupled, so that improving one generally comes at the expense of the other. A natural way to characterize this optimal trade-off is through the Pareto boundary, which identifies the set of operating points at which neither metric can be further improved without worsening the other. In the following, we investigate the corresponding Pareto boundary of the communication delay-control error performance pairs.

Specifically, the communication-control trade-off is captured by the performance pair $(\tau_U,V_\infty)$, where $\tau_U$ is determined by the downlink communication SINR and $V_\infty$ is governed by the closed-loop uplink-downlink control quality. Under a total downlink power budget, each beamforming design $(\mathbf{w}_D,\mathbf{w}_U)$ induces one feasible operating point in the $(\tau_U,V_\infty)$ plane. In this sense, it characterizes the fundamental delay-control trade-off of the JDCC system under shared wireless resources. By minimizing the communication delay under different target control-performance values, the Pareto boundary can be expressed as
\begin{equation}
	\mathcal{B}_{\mathrm{Pareto}} = \left\{\left(\tau_U^\star(V_\infty),V_\infty\right): V_\infty \in [V_\infty^{\min},\infty)\right\},
	\label{eq:pareto_boundary}
\end{equation}
where $\tau_U^\star(V_\infty)$ denotes the minimum achievable communication delay for a given steady-state control variance $V_\infty$, and $V_\infty^{\min}$ is the finite control-only benchmark in \eqref{eq:V_min}. Therefore, $\mathcal{B}_{\mathrm{Pareto}}$ provides a complete characterization of the best achievable communication-delay and control-error trade-off in the considered JDCC system.

\subsubsection{Pareto-Boundary Point Formulation}
\label{subsubsec:pareto_reformulation}

To characterize the Pareto boundary, we adopt a constraint-based formulation. Specifically, for a given target control-performance value $V_{\mathrm{th}}$, each boundary point can be obtained by minimizing the communication delay subject to the control-performance constraint and the total downlink power constraint, i.e.,
\begin{subequations}\label{pro:pareto}
	\label{eq:pareto_opt}
	\begin{align}
		\min_{\mathbf{w}_D,\, \mathbf{w}_U} \quad
		& \tau_U = \frac{Q_U}{B_{\mathrm{dn}}\log_2\!\left(1+\frac{|\mathbf{h}_U^H \mathbf{w}_U|^2}{|\mathbf{h}_U^H \mathbf{w}_D|^2 + \sigma_{\mathrm{dn}}^2}\right)}
		\label{eq:pareto_obj}                                           \\
		\text{s.t.} \quad
		& V_\infty \le V_{\mathrm{th}},
		\label{eq:pareto_const0}                                        \\
		& \|\mathbf{w}_D\|^2 + \|\mathbf{w}_U\|^2 \le P_{\mathrm{dn}}.
		\label{eq:pareto_const2}
	\end{align}
\end{subequations}
By sweeping $V_{\mathrm{th}}$ over the feasible range $[V_\infty^{\min},\infty)$, the complete Pareto boundary can be obtained.

Directly obtaining the minimum delay $\tau_U^\star(V_\infty)$ from \eqref{eq:pareto_opt} is highly challenging, since the communication objective and the control-performance constraint are both nonlinearly coupled through the beamforming vectors $(\mathbf{w}_D,\mathbf{w}_U)$. To facilitate the subsequent analysis, we further reformulate the boundary-point problem in the SINR domain. Since $\tau_U = Q_U/[B_{\mathrm{dn}}\log_2(1+\Gamma_U)]$ is strictly decreasing in the communication SINR, minimizing $\tau_U$ is equivalent to maximizing $\Gamma_U$. Meanwhile, the control-performance constraint $V_\infty \le V_{\mathrm{th}}$ can be converted into an equivalent downlink SINR requirement for the CD, as in \emph{Lemma \ref{lem:gamma_inv}}.

\begin{lemma}[Control Threshold SINR]
  \label{lem:gamma_inv}
    For a given variance $V_{\mathrm{th}} > \sigma_w^2$ and a fixed uplink quality satisfying $S_\alpha > |a|^2 V_{\mathrm{th}}/(V_{\mathrm{th}}-\sigma_w^2)$, the constraint $V_\infty \le V_{\mathrm{th}}$ is equivalent to $\Gamma_\alpha \ge \Gamma_\alpha^{\mathrm{th}}(V_{\mathrm{th}},S_\alpha)$ and  $\mathrm{SINR}_D^{\mathrm{dn}} \ge \gamma_D^{\mathrm{th}}(V_{\mathrm{th}},S_\alpha)$, where
    \begin{subequations}
    	\begin{align}
    		\Gamma_\alpha^{\mathrm{th}}(V_{\mathrm{th}},S_\alpha)
    		= \frac{|a|^2\,V_{\mathrm{th}}(S_\alpha - 1)}
    		{S_\alpha(V_{\mathrm{th}} - \sigma_w^2) - |a|^2 V_{\mathrm{th}}},
    		\label{eq:Gamma_alpha_th} \\
    		 \gamma_D^{\mathrm{th}}(V_{\mathrm{th}},S_\alpha)
    		= \bigl(\Gamma_\alpha^{\mathrm{th}}(V_{\mathrm{th}},S_\alpha)\bigr)^{1/\alpha_{\mathrm{dn}}} - 1.
    		\label{eq:gamma_d_req}
    	\end{align}
    \end{subequations}
\end{lemma}

\begin{IEEEproof}
  For fixed $S_\alpha$, the steady-state variance $V_\infty$ in \eqref{eq:V_infinity} is strictly decreasing in $\Gamma_\alpha$ over the stable regime. Therefore, the constraint $V_\infty \le V_{\mathrm{th}}$ is equivalent to $\Gamma_\alpha \ge \Gamma_\alpha^{\mathrm{th}}$, where $\Gamma_\alpha^{\mathrm{th}}$ is obtained by setting $V_\infty = V_{\mathrm{th}}$ in \eqref{eq:V_infinity} and solving for $\Gamma_\alpha$. The corresponding SINR threshold then follows from $\Gamma_\alpha = (1+\mathrm{SINR}_D^{\mathrm{dn}})^{\alpha_{\mathrm{dn}}}$.
\end{IEEEproof}

By Lemma~\ref{lem:gamma_inv}, the problem in \eqref{eq:pareto_opt} can be equivalently reformulated as
\begin{subequations}
  \label{eq:pareto_opt_sinr}
  \begin{align}
    \max_{\mathbf{w}_D,\, \mathbf{w}_U} \quad
     & \Gamma_U
    = \frac{|\mathbf{h}_U^H \mathbf{w}_U|^2}{|\mathbf{h}_U^H \mathbf{w}_D|^2 + \sigma_{\mathrm{dn}}^2}
    \label{eq:pareto_obj_sinr}                                                       \\
    \text{s.t.} \quad
     & \frac{|\mathbf{h}_D^H \mathbf{w}_D|^2}{|\mathbf{h}_D^H \mathbf{w}_U|^2 + \sigma_{\mathrm{dn}}^2}
    \ge \gamma_D^{\mathrm{th}}(V_{\mathrm{th}},S_\alpha),
    \label{eq:pareto_const1}                                                         \\
     & \|\mathbf{w}_D\|^2 + \|\mathbf{w}_U\|^2 \le P_{\mathrm{dn}}.
    \label{eq:pareto_const2_sinr}
  \end{align}
\end{subequations}
This reformulation shows that, for a given target control-performance value $V_{\mathrm{th}}$, the corresponding boundary point can be characterized through an equivalent downlink control-SINR requirement $\gamma_D^{\mathrm{th}}(V_{\mathrm{th}},S_\alpha)$. Therefore, instead of parameterizing the Pareto boundary directly by the control variance, it is more convenient to parameterize it by the target downlink control SINR. Let $\gamma_D$ denote such a target control SINR. Then, each feasible value of $\gamma_D$ specifies one boundary point through the solution of \eqref{eq:pareto_opt_sinr}. To identify the feasible range of $\gamma_D$, note that the smallest control SINR corresponds to the minimum requirement for closed-loop stability, which is given by $\gamma_D^{\min} = 
\bigl(|a|^2(S_\alpha-1)/(S_\alpha-|a|^2)\bigr)^{1/\alpha_{\mathrm{dn}}}-1$. On the other hand, the largest feasible control SINR is attained when all downlink power is allocated to the CD, i.e., $\gamma_D^{\max}  =  \frac{P_{\mathrm{dn}}\|\mathbf{h}_D\|^2}{\sigma_{\mathrm{dn}}^2}$. Therefore, for each $\gamma_D \in (\gamma_D^{\min}, \gamma_D^{\max}]$ solving \eqref{eq:pareto_opt_sinr} yields the optimal communication SINR $\Gamma_U^\star(\gamma_D)$, and the Pareto boundary can be equivalently rewritten as
\begin{equation}
	\begin{aligned}
		&\mathcal{B}_{\mathrm{Pareto}} =  \Bigg\{  \Big(\frac{Q_U}{B_{\mathrm{dn}}\log_2\Big(1+\Gamma_U^\star(\gamma_D)\Big)},\,V_\infty(\gamma_D)\Big) \\ & : \gamma_D \in \Big(\! \bigl(|a|^2(S_\alpha\!-\!1)/(S_\alpha\!-\!|a|^2)\bigr)^{1/\alpha_{\mathrm{dn}}}\!-\!1,\frac{P_{\mathrm{dn}}\|\mathbf{h}_D\|^2}{\sigma_{\mathrm{dn}}^2}\Big]\!\Bigg\}.
		\label{eq:pareto_boundary_parametric}
	\end{aligned}
\end{equation}

\subsubsection{Analytical Derivation of the Pareto Boundary}
\label{subsubsec:optimal_pareto_boundary}

To complete the characterization in \eqref{eq:pareto_boundary_parametric}, it remains to determine the optimal communication SINR $\Gamma_U^\star(\gamma_D)$ for each feasible control-SINR target $\gamma_D$. This requires solving the problem~\eqref{eq:pareto_opt_sinr}. Although this problem is still nonconvex, its optimal structure can be characterized through the associated Lagrangian function and the corresponding KKT conditions. Specifically, for a given $\gamma_D$, the Lagrangian function with multipliers $\lambda \ge 0$ and $\nu \ge 0$ is given by
\begin{align}
	\mathcal{L}(\mathbf{w}_D,\mathbf{w}_U,\lambda,\nu)
	\!= & \frac{|\mathbf{h}_U^H\mathbf{w}_U|^2}{|\mathbf{h}_U^H\mathbf{w}_D|^2\!\! +\! \sigma_{\mathrm{dn}}^2}
	\!+\! \lambda\!\left(\!\frac{|\mathbf{h}_D^H\mathbf{w}_D|^2}{|\mathbf{h}_D^H\mathbf{w}_U|^2 \!\!+\! \sigma_{\mathrm{dn}}^2}\! -\! \gamma_D\!\!\right) \notag \\
	& - \nu\bigl(\|\mathbf{w}_D\|^2 + \|\mathbf{w}_U\|^2 - P_{\mathrm{dn}}\bigr).
	\label{eq:lagrangian}
\end{align}
Based on the resulting first-order optimality conditions, the optimal communication SINR $\Gamma_U^\star(\gamma_D)$ can be analytically characterized as follows.

\begin{theorem}[Analytical Characterization of the Pareto Boundary]
	\label{thm:bf_structure}
For any feasible control-SINR target $\gamma_D \in (\gamma_D^{\min}, \gamma_D^{\max}]$, the corresponding point on the Pareto boundary is given by
\begin{equation}
			\left(
			\frac{Q_U}{B_{\mathrm{dn}}\log_2\!\bigl(1+\Gamma_U^\star(\gamma_D)\bigr)},
			\,V_\infty(\gamma_D)
			\right),
			\label{eq:pareto_point_gamma}
\end{equation}	
where $\Gamma_U^\star(\gamma_D)$ denotes the maximum achievable communication SINR under the control-SINR requirement $\gamma_D$, and $V_\infty(\gamma_D)$ is obtained from \eqref{eq:V_infinity} by setting $\Gamma_\alpha = (1+\gamma_D)^{\alpha_{\mathrm{dn}}}$. 	More specifically, $\Gamma_U^\star(\gamma_D) = \frac{|\mathbf{h}_U^H\mathbf{w}_U^\star|^2}{|\mathbf{h}_U^H\mathbf{w}_D^\star|^2 + \sigma_{\mathrm{dn}}^2}$ is achieved by the beamforming vectors:
\begin{align}
	\mathbf{w}_D^\star & = \sqrt{p_D}\,\frac{(\mathbf{I}_M + \mu_U \mathbf{h}_U \mathbf{h}_U^H)^{-1}\mathbf{h}_D}{\|(\mathbf{I}_M + \mu_U \mathbf{h}_U \mathbf{h}_U^H)^{-1}\mathbf{h}_D\|},
	\label{eq:wD_opt} \\
	\mathbf{w}_U^\star & = \sqrt{p_U}\,\frac{(\mathbf{I}_M + \mu_D \mathbf{h}_D \mathbf{h}_D^H)^{-1}\mathbf{h}_U}{\|(\mathbf{I}_M + \mu_D \mathbf{h}_D \mathbf{h}_D^H)^{-1}\mathbf{h}_U\|},
	\label{eq:wU_opt}
\end{align}
where the necessary optimality structure of Pareto-boundary satisfy
\begin{subequations}
	\label{eq:kkt_system}
	\begin{align}
		p_D + p_U & = P_{\mathrm{dn}},
		\label{eq:kkt_power} \\
		\frac{|\mathbf{h}_D^H \mathbf{w}_D^\star|^2}{|\mathbf{h}_D^H \mathbf{w}_U^\star|^2 + \sigma_{\mathrm{dn}}^2}& = \gamma_D,
		\label{eq:kkt_sinr} \\
		\mu_U & = \frac{1}{\nu^\star}\frac{\Gamma_U^\star(\gamma_D)}{|\mathbf{h}_U^H\mathbf{w}_D^\star|^2 + \sigma_{\mathrm{dn}}^2},
		\label{eq:kkt_muU} \\
		\mu_D & = \frac{1}{\nu^\star}\frac{\lambda^\star\gamma_D}{|\mathbf{h}_D^H\mathbf{w}_U^\star|^2 + \sigma_{\mathrm{dn}}^2}.
		\label{eq:kkt_muD}
	\end{align}
\end{subequations}
\end{theorem}

\begin{IEEEproof}
	Please refer to Appendix~\ref{app:proof_bf_structure}.
\end{IEEEproof}

\subsection{JDCC Performance Regions Under MRT and ZF}
\label{subsec:linear_bf}

The Pareto boundary characterized in the previous subsection reveals the fundamental trade-off of JDCC systems, but its evaluation still requires solving the associated optimization problem and therefore does not yield simple closed-form expressions. To gain more explicit analytical insight, we next turn to two typical linear beamforming schemes, namely, MRT and ZF. Under these two schemes, the resulting delay-control performance regions can be characterized in closed form.

\subsubsection{Performance region under MRT}
\label{subsubsec:mrt}

Under the MRT scheme, the communication and control beams are aligned with the intended CU and CD channels, respectively, i.e.,
\begin{equation}
  \mathbf{w}_D^{\mathrm{MRT}} = \sqrt{p_D^{\mathrm{MRT}}}\;
  \frac{\mathbf{h}_D}{\|\mathbf{h}_D\|},
  \
  \mathbf{w}_U^{\mathrm{MRT}} = \sqrt{p_U^{\mathrm{MRT}}}\;
  \frac{\mathbf{h}_U}{\|\mathbf{h}_U\|}.
  \label{eq:bf_mrt}
\end{equation}
where $p_D^{\mathrm{MRT}}$ and $p_U^{\mathrm{MRT}}$ denote the power allocated to the CD and the CU, respectively. The following theorem provides the closed-form performance region under MRT.

\begin{theorem}
  \label{thm:mrt_pareto}
  Under the MRT scheme~\eqref{eq:bf_mrt}, the JDCC performance region is given by
  \begin{equation}
    \begin{aligned}
       & \mathcal{B}_{\mathrm{MRT}} =   \Bigg\{\bigg(\frac{Q_U}{B_{\mathrm{dn}}\log_2 \left(1
         +\! \frac{p_U^{\mathrm{MRT}}\,\|\mathbf{h}_U\|^2}
      {p_D^{\mathrm{MRT}} \rho \|\mathbf{h}_U\|^2 + \sigma_{\mathrm{dn}}^2} \!\right)}\, , \\ & \quad\frac{\sigma_w^2\,S_\alpha\,(1+\gamma_D)^{\alpha_{\mathrm{dn}}}}{S_\alpha(1+\gamma_D)^{\alpha_{\mathrm{dn}}}-|a|^2(S_\alpha+(1+\gamma_D)^{\alpha_{\mathrm{dn}}}-1)}\bigg)                                       \\
       & \quad : \gamma_D \!\in\! \Big(\! \bigl(|a|^2(S_\alpha\!-\!1)/(S_\alpha\!-\!|a|^2)\bigr)^{1/\alpha_{\mathrm{dn}}}\!\!-\!1,\frac{P_{\mathrm{dn}}\|\mathbf{h}_D\|^2}{\sigma_{\mathrm{dn}}^2}\Big]\Bigg\},
    \end{aligned}
  \end{equation}
  where the optimal power allocation is
  \begin{equation}
    p_D^{\mathrm{MRT}}
    = \frac{\gamma_D\bigl(P_{\mathrm{dn}}\,\rho\,\|\mathbf{h}_D\|^2 + \sigma_{\mathrm{dn}}^2\bigr)}
    {\|\mathbf{h}_D\|^2\bigl(1 + \gamma_D\,\rho\bigr)}, \
    p_U^{\mathrm{MRT}} = P_{\mathrm{dn}} - p_D^{\mathrm{MRT}},
    \label{eq:pD_mrt}
  \end{equation}
  and $\rho =  	\frac{|\mathbf{h}_D^H \mathbf{h}_U|^2}
  {\|\mathbf{h}_D\|^2\,\|\mathbf{h}_U\|^2}
  \in [0,1] $ denotes the channel correlation coefficient.
\end{theorem}

\begin{IEEEproof}
	Please refer to Appendix~\ref{app:proof_mrt_pareto}.
\end{IEEEproof}

\subsubsection{Performance Region Under ZF}
\label{subsubsec:zf}

Under the ZF scheme, the communication and control beams are projected onto the null spaces of the CD and CU channels, respectively, i.e.,
\begin{equation}
	\mathbf{w}_D^{\mathrm{ZF}} = \sqrt{p_D^{\mathrm{ZF}}}\,\frac{\mathbf{P}_U^{\perp}\mathbf{h}_D}{\|\mathbf{P}_U^{\perp}\mathbf{h}_D\|},
	\
	\mathbf{w}_U^{\mathrm{ZF}} = \sqrt{p_U^{\mathrm{ZF}}}\,\frac{\mathbf{P}_D^{\perp}\mathbf{h}_U}{\|\mathbf{P}_D^{\perp}\mathbf{h}_U\|},
	\label{eq:zf_directions}
\end{equation}
where
\begin{equation}
	\mathbf{P}_U^{\perp} \triangleq \mathbf{I}_M - \frac{\mathbf{h}_U\mathbf{h}_U^H}{\|\mathbf{h}_U\|^2},
	\qquad
	\mathbf{P}_D^{\perp} \triangleq \mathbf{I}_M - \frac{\mathbf{h}_D\mathbf{h}_D^H}{\|\mathbf{h}_D\|^2},
	\label{eq:proj_def}
\end{equation}
and $p_D^{\mathrm{ZF}}$ and $p_U^{\mathrm{ZF}}$ denote the power allocated to the CD and the CU, respectively. The following theorem provides the closed-form performance region under ZF.

\begin{theorem}
	\label{thm:zf_pareto}
 	 Under the ZF beamforming scheme \eqref{eq:zf_directions}, the JDCC performance region is given by
	\begin{equation}
		\begin{aligned}
			& \mathcal{B}_{\mathrm{ZF}}   = \Bigg\{ \bigg(\frac{Q_U}{B_{\mathrm{dn}}\log_2\left(1+\frac{p_U^{\mathrm{ZF}}\|\mathbf{h}_U\|^2(1-\rho)}{\sigma_{\mathrm{dn}}^2}\right)},\, \\ &  \frac{\sigma_w^2\,S_\alpha\,(1+\gamma_D)^{\alpha_{\mathrm{dn}}}}{S_\alpha(1+\gamma_D)^{\alpha_{\mathrm{dn}}}-|a|^2(S_\alpha+(1+\gamma_D)^{\alpha_{\mathrm{dn}}}-1)} \bigg) \\
			&: \! \gamma_D \!\!\in\!\! \Big(\! \bigl(|a|^2\!(S_\alpha\!-\!1)/(S_\alpha\!-\!|a|^2)\bigr)^{1/\alpha_{\mathrm{dn}}}\!\!-\!1,\!\frac{(1\!-\!\rho)P_{\mathrm{dn}}\!\|\mathbf{h}_D\|^2}{\sigma_{\mathrm{dn}}^2}\!\Big]\!\Bigg\}\!,
		\end{aligned}
	\end{equation}
	where
	\begin{equation}
		p_U^{\mathrm{ZF}} = P_{\mathrm{dn}} -\frac{\gamma_D\,\sigma_{\mathrm{dn}}^2}
		{\|\mathbf{h}_D\|^2(1-\rho)}.
		\label{eq:pU_zf}
	\end{equation}
\end{theorem}

\begin{IEEEproof}
	Please refer to Appendix~\ref{app:proof_zf_pareto}.
\end{IEEEproof}

Based on the closed-form JDCC performance regions under MRT and ZF beamforming, we next compare these two scheme-induced trade-off curves with the Pareto boundary.

\begin{remark}
	\label{rem:pareto_mrt_zf_comparison}
	For any feasible control-SINR $\gamma_D$, the Pareto boundary is obtained by optimizing over all feasible beamforming designs, whereas MRT and ZF correspond to two specific beamforming schemes. Therefore, the MRT and ZF curves lie on or outside the Pareto boundary.
	
	$\bullet$ When $\rho = 0$, the CU and CD channels become orthogonal, and thus MRT and ZF coincide and both achieve the global Pareto boundary.
	
	$\bullet$ When $\rho > 0$, the relative performance of MRT and ZF depends on the available downlink power $P_{\mathrm{dn}}$. In particular, within the interval:
	\begin{equation}
		\frac{\gamma_D \sigma_{\mathrm{dn}}^2}{\|\mathbf{h}_D\|^2}
		\le
		P_{\mathrm{dn}}
		<
		\frac{\gamma_D \sigma_{\mathrm{dn}}^2}{\|\mathbf{h}_D\|^2(1-\rho)},
		\label{eq:mrt_only_region_remark}
	\end{equation}
	ZF is infeasible, whereas MRT remains feasible. When both schemes are feasible, their power threshold is given by
	\begin{equation}
		P_{\mathrm{dn}}^{\mathrm{MRT/ZF}}
		=
		\frac{-c_1 + \sqrt{c_1^2 - 4 c_2 c_0}}{2 c_2},
	\end{equation}
	where
	\begin{align*}
	& c_2 \! = \gamma_D \rho^2 (1-\rho) \|\mathbf{h}_D\|^4 \|\mathbf{h}_U\|^2,                                                                                                                           \\
	& c_1 \! =\! \!\sigma_{\mathrm{dn}}^2\! \|\mathbf{h}_D\|^2\! \rho \bigl[\! \gamma_D \|\mathbf{h}_U\|^2 (1\! -\! \rho \!-\! \gamma_D\rho) \!\!+\!\! \|\mathbf{h}_D\|^2 (\gamma_D\! -\! \gamma_D\rho \!-\! 1)\! \bigr]\!, \\
	& c_0 \! = - \gamma_D^2 (\sigma_{\mathrm{dn}}^2)^2 \rho (\|\mathbf{h}_D\|^2 + \|\mathbf{h}_U\|^2).
	\end{align*}
	Specifically, MRT is better when $P_{\mathrm{dn}} < P_{\mathrm{dn}}^{\mathrm{MRT/ZF}}$, whereas ZF is better when $P_{\mathrm{dn}} > P_{\mathrm{dn}}^{\mathrm{MRT/ZF}}$.
\end{remark}

\section{Outage Analysis of JDCC Systems}
\label{sec:outage}

The Pareto boundary characterized in the previous section explicitly reveals the optimal communication-control trade-off of the dual-functional JDCC system. This section further investigates the corresponding outage behavior to evaluate the reliability of dual-functional JDCC operation. To capture the effect of instantaneous CSI on the outage probability, the channel vectors are modeled as $\mathbf{h}_D \sim \mathcal{CN}(\mathbf{0},\beta_D\mathbf{I}_M)$ and $\mathbf{h}_U \sim \mathcal{CN}(\mathbf{0},\beta_U\mathbf{I}_M)$, where $\beta_D > 0$ and $\beta_U > 0$ denote the large-scale path-loss coefficients. Based on this model, we first characterize the outage probabilities of the two single-function benchmark paradigms, and then extend the analysis to the joint outage probability of the dual-functional system.

\subsection{Single-Function Outage Probability}
\label{subsec:single_outage}

We first characterize the outage probabilities of the two single-function benchmark paradigms, namely the communication-only and control-only schemes. These two benchmark results serve as basic reliability limits and provide a useful foundation for the subsequent joint outage analysis of the dual-functional JDCC system.

\subsubsection{Paradigm A: Communication-Only Outage Probability}

Under Paradigm A, the BS adopts the communication-only MRT scheme in \eqref{eq:mrt_comm}. For a given communication-delay threshold $\tau_{\mathrm{req}} > 0$, the communication-only outage probability is defined as
\begin{equation}
	P_{\mathrm{out}}^{\mathrm{com}}
	\triangleq
	\Pr\!\left(\tau_U^{\mathrm{com}} > \tau_{\mathrm{req}}\right),
	\label{eq:comm_outage_def}
\end{equation}
where $\tau_U^{\mathrm{com}}$ is given in \eqref{eq:Du_min}. Let $G_U  = \|\mathbf{h}_U\|^2/\beta_U \sim \mathrm{Gamma}(M,1)$. Then, the outage event is equivalent to $G_U < \eta_U$, where $\eta_U = \frac{\gamma_U^{\mathrm{req}}\sigma_{\mathrm{dn}}^2}{P_{\mathrm{dn}}\beta_U}$ and $\gamma_U^{\mathrm{req}} = 2^{Q_U/(B_{\mathrm{dn}}\tau_{\mathrm{req}})} - 1$. Hence,
\begin{equation}
	P_{\mathrm{out}}^{\mathrm{com}}
	=
	F_G(\eta_U)
	=
	1 - e^{-\eta_U}\sum_{k=0}^{M-1}\frac{\eta_U^k}{k!}.
	\label{eq:comm_outage}
\end{equation}

\subsubsection{Paradigm B: Control-Only Outage Probability}
\label{subsec:ctrl_outage}

Under Paradigm B, the BS adopts the control-only MRT scheme in \eqref{eq:mrt_ctrl}. For a given control-error threshold $V_{\mathrm{req}} > \sigma_w^2$, the control-only outage probability is defined as
\begin{equation}
	P_{\mathrm{out}}^{\mathrm{ctrl}}
	\triangleq
	\Pr\!\left(V_\infty \ge V_{\mathrm{req}}\right).
	\label{eq:ctrl_outage_vth}
\end{equation}
Using the steady-state variance expression in \eqref{eq:V_infinity}, this outage event can be equivalently written as
\begin{equation}
	P_{\mathrm{out}}^{\mathrm{ctrl}}
	=
	1 - \Pr\!\left(S_\alpha > \frac{|a|^2V_{\mathrm{req}}}{V_{\mathrm{req}}-\sigma_w^2},\ \widetilde{\gamma}_D \ge \gamma_D^{\mathrm{req}}(V_{\mathrm{req}},S_\alpha)\right),
	\label{eq:ctrl_outage_uplink_downlink}
\end{equation}
where $\widetilde{\gamma}_D$ denotes the instantaneous downlink control SINR and $\gamma_D^{\mathrm{req}}(V_{\mathrm{req}},S_\alpha)$ is the same expression given in \eqref{eq:gamma_d_req}. However, this outage probability depends jointly on the uplink state-reporting quality and the downlink control-delivery quality. In particular, the corresponding uplink quality $S_\alpha$ and instantaneous downlink control SINR $\widetilde{\gamma}_D$ are coupled through the same channel realization $\|\mathbf{h}_D\|^2$. Therefore, to characterize the resulting joint outage event, it is convenient to introduce the normalized channel gain:
\begin{equation}
	G_D = \frac{\|\mathbf{h}_D\|^2}{\beta_D} \sim \mathrm{Gamma}(M,1).\end{equation}
In this way, the two coupled quantities in \eqref{eq:ctrl_outage_uplink_downlink} can be rewritten in terms of $G_D$ as 
\begin{equation}
	S_\alpha = \left(1 + \frac{P_{\mathrm{up}}\beta_D}{\sigma_{\mathrm{up}}^2}G_D\right)^{\alpha_{\mathrm{up}}}, \ \widetilde{\gamma}_D
	=
	\frac{P_{\mathrm{dn}}\beta_D}{\sigma_{\mathrm{dn}}^2}G_D.
\end{equation}
Note that the required downlink control SINR $\gamma_D^{\mathrm{req}}(V_{\mathrm{req}},S_\alpha)$ in \eqref{eq:gamma_d_req} also depends on $S_\alpha$, and is therefore implicitly determined by the same random variable $G_D$. Accordingly, the outage event in \eqref{eq:ctrl_outage_uplink_downlink} is completely determined by the single random variable $G_D$, such as
\begin{equation}
	P_{\mathrm{out}}^{\mathrm{ctrl}}
	=
	1 - \Pr\!\left(G_D > \eta_V,\ \bar{\gamma}_d G_D \ge \gamma_D^{\mathrm{req}}(G_D)\right),
	\label{eq:ctrl_outage_general_event}
\end{equation}
where $\eta_V = \frac{1}{\bar{\gamma}_{\mathrm{up}}}
\left[\left( \frac{|a|^2V_{\mathrm{req}}}{V_{\mathrm{req}}-\sigma_w^2}
\right)^{1/\alpha_{\mathrm{up}}} - 1
\right]$, $\gamma_D^{\mathrm{req}}(G_D)
=
\left(
\frac{|a|^2V_{\mathrm{req}}(S_\alpha-1)}
{S_\alpha(V_{\mathrm{req}}-\sigma_w^2)-|a|^2V_{\mathrm{req}}}
\right)^{1/\alpha_{\mathrm{dn}}}
- 1$, $S_\alpha = (1 + \bar{\gamma}_{\mathrm{up}} G_D )^{\alpha_{\mathrm{up}}}$, $\bar{\gamma}_{\mathrm{up}} = \frac{P_{\mathrm{up}}\beta_D}{\sigma_{\mathrm{up}}^2} $, and $\bar{\gamma}_d =  \frac{P_{\mathrm{dn}}\beta_D}{\sigma_{\mathrm{dn}}^2}$.

Therefore, the key step is to identify the channel-gain threshold above which both the uplink feasibility condition and the downlink control-SINR requirement are simultaneously satisfied. Based on this observation, the control-only outage probability under MRT can be characterized as follows.

\begin{theorem}[Control-Only Outage under MRT]
	\label{thm:ctrl_outage_mrt}
	The control-only outage probability under MRT is given by
	\begin{equation}
		P_{\mathrm{out}}^{\mathrm{ctrl}}
		=
		1 - e^{-\eta_{\mathrm{ctrl}}}\sum_{k=0}^{M-1}\frac{\eta_{\mathrm{ctrl}}^k}{k!},
		\label{eq:ctrl_outage_mrt}
	\end{equation}
	where $\eta_{\mathrm{ctrl}} \in (\eta_V,\infty)$ is the unique solution to
	\begin{equation}
		\bar{\gamma}_d \eta_{\mathrm{ctrl}} = \gamma_D^{\mathrm{req}}(\eta_{\mathrm{ctrl}}).
		\label{eq:eta_ctrl_eq}
	\end{equation}
\end{theorem}

\begin{IEEEproof}
	From \eqref{eq:ctrl_outage_general_event}, the success event under control-only MRT is $G_D > \eta_V$ and $\bar{\gamma}_d G_D \ge \gamma_D^{\mathrm{req}}(G_D)$, which can be rewritten as, for $x>\eta_V$,
	\begin{equation}
		\gamma_D^{\mathrm{req}}(x)
		=
		\left(
		\frac{|a|^2V_{\mathrm{req}}\bigl((1 + \bar{\gamma}_{\mathrm{up}}x)^{\alpha_{\mathrm{up}}}-1\bigr)}
		{(1 + \bar{\gamma}_{\mathrm{up}}x)^{\alpha_{\mathrm{up}}}(V_{\mathrm{req}}-\sigma_w^2)-|a|^2V_{\mathrm{req}}}
		\right)^{1/\alpha_{\mathrm{dn}}}
		- 1,
	\end{equation}
	It can be verified that $\gamma_D^{\mathrm{req}}(x)$ is continuous and strictly decreasing for $x>\eta_V$, whereas $\bar{\gamma}_d x$ is strictly increasing in $x$. Moreover, $\gamma_D^{\mathrm{req}}(x)\to\infty$ as $x\to\eta_V^{+}$, while $\gamma_D^{\mathrm{req}}(x)$ approaches the finite lower limit corresponding to the minimum downlink SINR required for control stability as $x\to\infty$. Therefore, the equation $\bar{\gamma}_d x = \gamma_D^{\mathrm{req}}(x)$	admits a unique solution $\eta_{\mathrm{ctrl}} \in (\eta_V,\infty)$. Hence, the success event reduces to $G_D \ge \eta_{\mathrm{ctrl}}$, and the control-only outage probability is obtained from the CDF of $G_D \sim \mathrm{Gamma}(M,1)$ as $P_{\mathrm{out}}^{\mathrm{ctrl}}
	=
	\Pr(G_D < \eta_{\mathrm{ctrl}})
	=
	F_G(\eta_{\mathrm{ctrl}})
	=
	1 - e^{-\eta_{\mathrm{ctrl}}}\sum_{k=0}^{M-1}\frac{\eta_{\mathrm{ctrl}}^k}{k!}$.
\end{IEEEproof}

\subsection{JDCC Outage Probability}
\label{subsec:joint_outage}

After characterizing the outage probabilities of the two single-function benchmark paradigms, we next study the outage behavior of the dual-functional JDCC system. Different from the single-function cases, the JDCC outage event is jointly constrained by the communication and control requirements. Therefore, we use their joint probability \cite{gan2024coverage} to quantify the reliability of supporting communication and control simultaneously. Specifically, for given thresholds $\tau_{\mathrm{req}} > 0$ and $V_{\mathrm{req}} > \sigma_w^2$, the JDCC outage probability is defined as
\begin{equation}
	P_{\mathrm{out}}^{\mathrm{JDCC}}
	\triangleq
	1 - \Pr\!\left(\tau_U \le \tau_{\mathrm{req}},\ V_\infty \le V_{\mathrm{req}}\right),
	\label{eq:joint_outage_rate_variance}
\end{equation}
where $\tau_{\mathrm{req}} > 0 $ is the maximum tolerable communication transmission delay and $V_{\mathrm{req}} > \sigma_w^2$ is the maximum tolerable steady-state control error variance. On the other hand, the non-outage event in \eqref{eq:joint_outage_rate_variance} can be equivalently rewritten as
\begin{equation}
	P_{\mathrm{out}}^{\mathrm{JDCC}}
	=
	1 - \Pr\!\left(\mathcal{E}_{\mathrm{suc}}^{\mathrm{JDCC}}\right),
	\label{eq:joint_outage_def}
\end{equation}
where
\begin{equation}
	\mathcal{E}_{\mathrm{suc}}^{\mathrm{JDCC}}
	\!=\!
	\left\{
	\mathrm{SINR}_U^{\mathrm{dn}} \! \ge \! \gamma_U^{\mathrm{req}},\
	G_D \! >\! \eta_V,\
	\mathrm{SINR}_D^{\mathrm{dn}} \!\ge\! \gamma_D^{\mathrm{req}}(G_D)
	\right\}.
	\label{eq:joint_success_def}
\end{equation}
Here, the first condition corresponds to the communication non-outage event, while the latter two jointly characterize the control non-outage event. In the following, we characterize the JDCC outage probability under MRT and ZF, respectively.

\subsubsection{JDCC Outage Probability under MRT}

We first specialize~\eqref{eq:joint_success_def} to MRT. The corresponding joint outage characterization is given next.

\begin{theorem}[Joint JDCC Outage Probability under MRT]
  \label{thm:icac_mrt}
  The JDCC outage probability under MRT is given by
  \begin{align}
    P_{\mathrm{out}}^{\mathrm{JDCC,MRT}} = 1 - & \int_{\eta_V}^{\infty}\!\int_0^{r_{\max}(x)} f_G(x)\,f_\rho(r) \notag \\
                                                & \times \mathbf{1}\!\left(x > \xi(x,r)\right)\bigl[1-F_G\!\bigl(\psi(x,r)\bigr)\bigr]\,\mathrm{d}r\,\mathrm{d}x,
    \label{eq:icac_mrt_out}
  \end{align}
  where $\mathbf{1}(\cdot)$ denotes the indicator function, $f_\rho(r)
  =
  (M-1)(1-r)^{M-2}$, $f_G(x) = \frac{x^{M-1}e^{-x}}{(M-1)!}$, $F_G(x) = \frac{\gamma(M,x)}{\Gamma(M)}
  = 1 - e^{-x}\sum_{k=0}^{M-1}\frac{x^k}{k!}$, and, for each $x > \eta_V$,
  \begin{subequations}
  	\begin{align}
  		r_{\max}(x)
  		& =
  		\min\!\left(1,\,\frac{1}{\sqrt{\gamma_U^{\mathrm{req}}\gamma_D^{\mathrm{req}}(x)}}\right),
  		\label{eq:r_max_x} \\
  		\xi(x,r)
  		& =
  		\frac{\eta_f(x)\bigl(1+\gamma_U^{\mathrm{req}}r\bigr)}
  		{1-\gamma_U^{\mathrm{req}}\gamma_D^{\mathrm{req}}(x)r^2},
  		\label{eq:xi_mrt} \\
  		\psi(x,r)
  		& =
  		\frac{\eta_U x\bigl(1+\gamma_D^{\mathrm{req}}(x)r\bigr)}
  		{x\bigl(1-\gamma_U^{\mathrm{req}}\gamma_D^{\mathrm{req}}(x)r^2\bigr)-\eta_f(x)\bigl(1+\gamma_U^{\mathrm{req}}r\bigr)},
  		\label{eq:psi_mrt} \\
		\eta_f(x)
  		& =
  		\frac{\gamma_D^{\mathrm{req}}(x)}{\bar{\gamma}_d}.
  		\label{eq:eta_f_zf}
  	\end{align}
  \end{subequations}
\end{theorem}

\begin{IEEEproof}
	Please refer to Appendix~\ref{app:proof_joint_mrt_outage}.
\end{IEEEproof}

\subsubsection{Joint JDCC Outage under ZF}

We next specialize~\eqref{eq:joint_success_def} to ZF. The corresponding characterization is stated below.

\begin{theorem}[Joint JDCC Outage under ZF]
  \label{thm:icac_zf}
  The joint JDCC outage probability under ZF is
  \begin{align}
    P_{\mathrm{out}}^{\mathrm{JDCC,ZF}} = 1 - & \int_{\eta_V}^{\infty}\int_0^{[1-\eta_f(x)/x]^+} f_G(x)\,f_\rho(r) \notag \\
                                              & \times \bigl[1-F_G\bigl(\phi(x,r)\bigr)\bigr]\,\mathrm{d}r\,\mathrm{d}x,
    \label{eq:icac_zf_out}
  \end{align}
  where $[z]^+ = \max(z,0)$ and, for each $x > \eta_V$, and $\phi(x,r) = \frac{\eta_U x}{(1-r)x-\eta_f(x)}$.
\end{theorem}

\begin{IEEEproof}
	Please refer to Appendix~\ref{app:proof_joint_zf_outage}.
\end{IEEEproof}

\section{Numerical Results}
\label{sec:numerical}
In this section, numerical results are provided to validate the theoretical
results and evaluate the performance of JDCC systems. We simulate the
large-scale fading coefficients of $\mathbf{h}_U$ and $\mathbf{h}_D$ as
$\beta_k=C_0 d_k^{-\alpha_{\mathrm{pl}}}$ for $k\in\{U,D\}$, where the
reference gain is $C_0=-30$ dB and the path-loss exponent is
$\alpha_{\mathrm{pl}}=3.2$. The CU and CD are located at distances $d_U=100$ m
and $d_D=120$ m from the BS, respectively. We consider a BS equipped with $M=4$ antennas, with transmit
powers $P_{\mathrm{dn}}=P_{\mathrm{up}}=-30$ dBm for the uplink and downlink. Specifically, for the control process, the plant coefficient is set to $a=1.2+1.2i$, $b=1$, and the process-noise variance is
$\sigma_w^2=10^{-2}$. The uplink CD state-reporting bandwidth is set to
$B_{\mathrm{up}}=10$ kHz, whereas the downlink bandwidth is set to
$B_{\mathrm{dn}}=20$ kHz. The control sampling period is
set to $T_s^D = 0.1$ ms. The receiver noise
power spectral density is set to $N_0=-174$ dBm/Hz. The communication payload size is set to
$Q_U = 1000$ bits.

Fig.~\ref{fig0_control_variance} plots the evolution of the normalized state variance $V_n/\sigma_w^2$ over the control interval index $n$. The Monte Carlo markers closely match the theoretical curves at each interval, which verifies the derived state-variance evolution. Under $P_{\mathrm{dn}} = P_{\mathrm{up}} = -30$ dBm, the case $a=1.2+1.2i$ converges to the steady-state variance in \emph{Theorem~1} within a few intervals, whereas the case $a=5+6i$ becomes unstable and its variance keeps increasing. After increasing the transmit powers to $P_{\mathrm{dn}} = P_{\mathrm{up}} = 0$ dBm, even the larger-$|a|$ case converges to a finite steady-state variance. These results verify \emph{Theorem~1} and show the critical role of communication link quality in maintaining closed-loop stability.

\begin{figure}[t]
	\centering
	\includegraphics[width=0.85\linewidth]{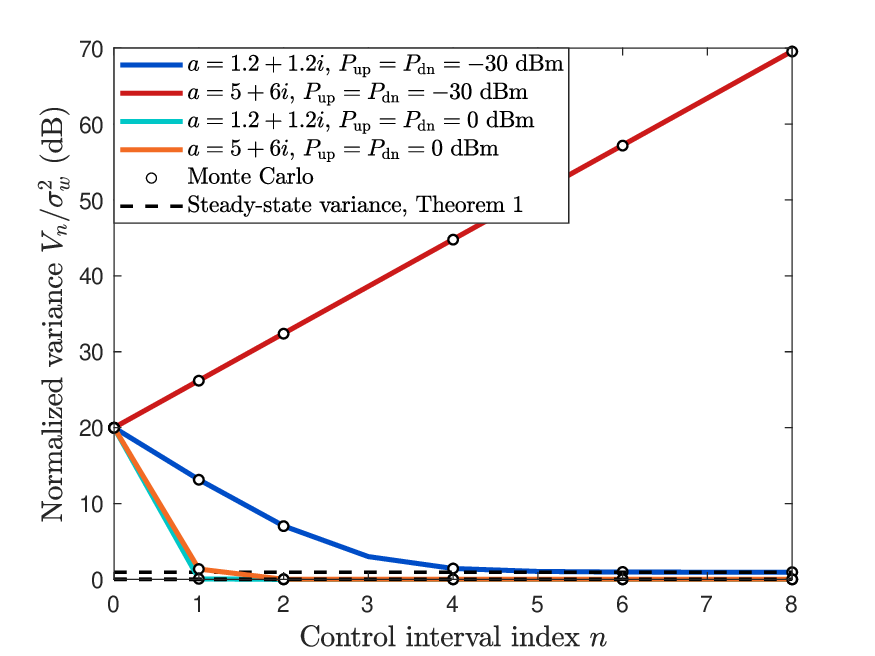}\vspace{-2mm}
	\caption{Normalized state variance $V_n/\sigma_w^2$ versus control interval index $n$, where the initial variance is $V_0=1$.}
	\label{fig0_control_variance}
\end{figure}

\begin{figure}[t]
  \centering
  \includegraphics[width=0.85\linewidth]{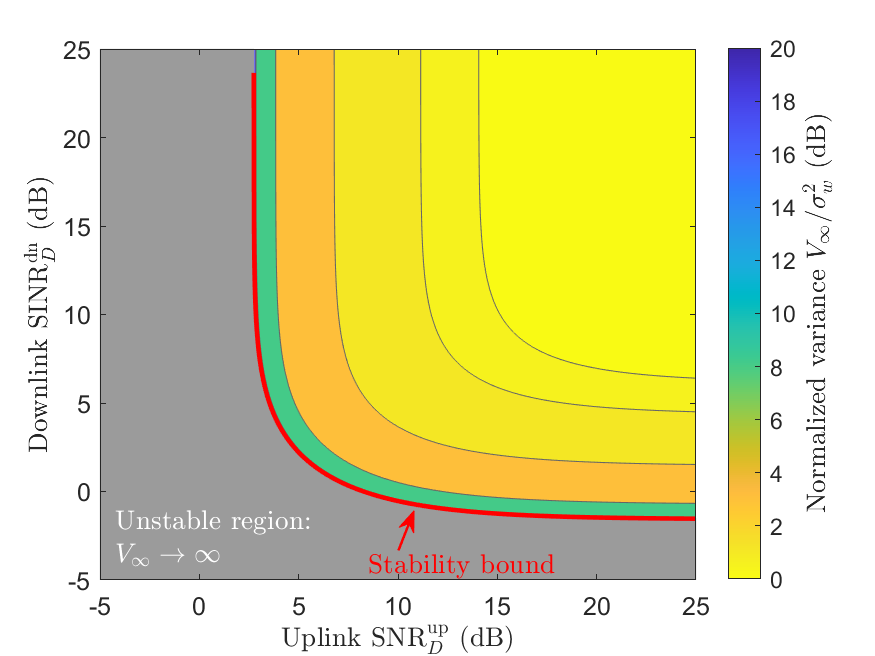}\vspace{-2mm}

  \caption{Normalized state variance $V_\infty/\sigma_w^2$ versus uplink $\mathrm{SNR}_D^{\mathrm{up}}$ and downlink $\mathrm{SINR}_D^{\mathrm{dn}}$.}
  \label{fig1_variance_bound}
\end{figure}

Fig.~\ref{fig1_variance_bound} illustrates the deterministic role of communication link quality in maintaining closed-loop stability. The red curve, obtained from the stability conditions in \emph{Theorem~1}, serves as the stability bound in the uplink-downlink SNR plane and clearly separates the stable and unstable regions. Below this bound, the wireless link quality is insufficient to support stable control. Within the stable region, the normalized variance decreases markedly as either $\mathrm{SNR}_D^{\mathrm{up}}$ or $\mathrm{SINR}_D^{\mathrm{dn}}$ increases, showing that both uplink state reporting and downlink control delivery affect the final control accuracy. This result further confirms that control stability is jointly determined by the uplink and downlink qualities rather than by either link alone.

\begin{figure}[t]
  \centering
  \includegraphics[width=0.85\linewidth]{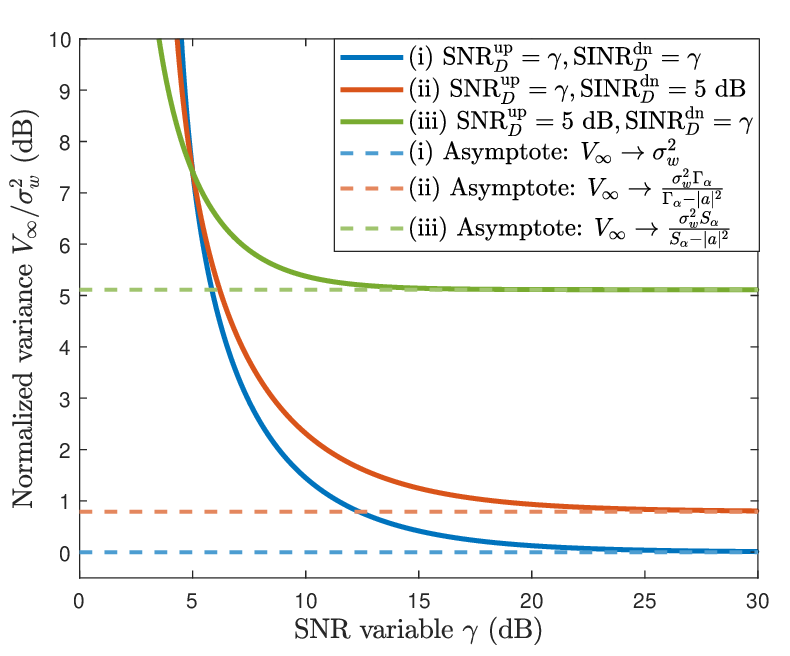}\vspace{-2mm}

  \caption{Normalized variance $V_\infty/\sigma_w^2$ for three asymptotic regimes of $\mathrm{SNR}_D^{\mathrm{up}}$ and $\mathrm{SINR}_D^{\mathrm{dn}}$.}
  \label{fig2_variance_asymptotics}
\end{figure}

Fig.~\ref{fig2_variance_asymptotics} examines the asymptotic behavior of the normalized variance $V_\infty/\sigma_w^2$ under the three high-SNR regimes in Corollary~\ref{cor:asymptotic}. As $\gamma$ increases, the exact results in all three cases gradually approach their corresponding asymptotic values, which verifies Corollary~\ref{cor:asymptotic}. In case (i), where both uplink and downlink qualities increase simultaneously, the normalized variance approaches $0$ dB, indicating asymptotically vanishing control error. By contrast, in cases (ii) and (iii), where only one link improves while the other remains fixed, the variance converges to a non-zero floor determined by the limited link.

\begin{figure}[t]
  \centering
  \includegraphics[width=0.85\linewidth]{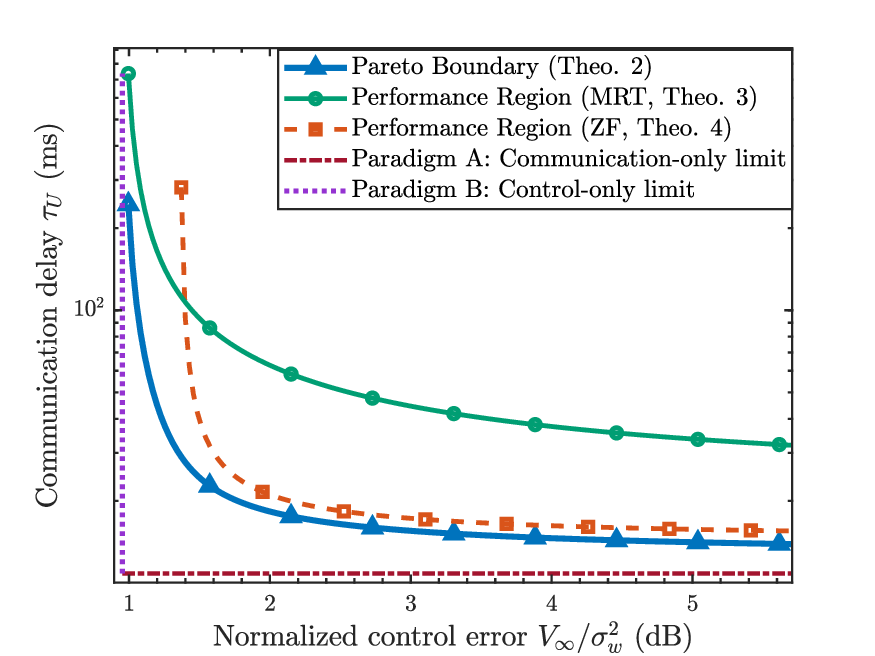}\vspace{-2mm}

  \caption{Performance region of $(\tau_U,V_\infty)$ under Pareto-optimal, MRT, and ZF beamforming under $\rho=0.5$.}
  \label{fig3_pareto_boundary}
\end{figure}

Fig.~\ref{fig3_pareto_boundary} illustrates the Pareto boundary of the considered JDCC system, together with the achievable delay-control performance regions under MRT and ZF, as well as the communication-only and control-only benchmark limits. It is observed that the Pareto boundary serves as the outer boundary of the MRT- and ZF-based achievable regions, and thus characterizes the optimal communication delay-control error trade-off. In addition, when a sufficiently large communication delay is allowed, the Pareto-optimal control performance approaches the control-only limit. However, even when a large control error is tolerated, the Pareto-optimal communication delay still cannot approach the communication-only limit. This is because closed-loop stability always requires a non-zero amount of power to support the control link, and hence any finite control-error target inevitably consumes part of the communication resource.

\begin{figure}[t]
  \centering
  \includegraphics[width=0.85\linewidth]{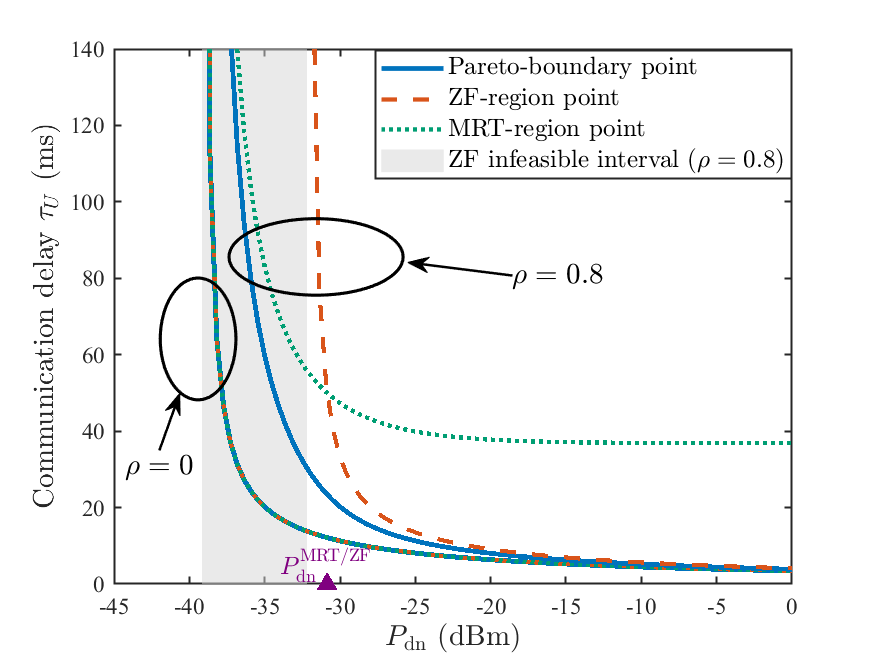}\vspace{-2mm}

  \caption{Communication delay $\tau_U$ (s) versus total downlink power $P_{\mathrm{dn}}$ with fixed control-SINR target $\gamma_D=0$ dB.}
  \label{fig4_mrt_zf}
\end{figure}

Fig.~\ref{fig4_mrt_zf} further compares the communication delay achieved by the Pareto-boundary solution, MRT, and ZF versus the total downlink transmit power $P_{\mathrm{dn}}$ under the fixed control-SINR target $\gamma_D=0$ dB. When $\rho=0$, the CU and CD channels are orthogonal, and the Pareto boundary coincides with the MRT and ZF regions over the entire power range, which verifies Remark~\ref{rem:pareto_mrt_zf_comparison}. When $\rho>0$, ZF exhibits an infeasible low-power region due to the projection loss caused by interference suppression. As $P_{\mathrm{dn}}$ increases, however, the interference-mitigation gain of ZF gradually outweighs the array-gain advantage of MRT. In particular, once $P_{\mathrm{dn}} > P_{\mathrm{dn}}^{\mathrm{MRT/ZF}}$, ZF achieves a lower communication delay than MRT, which is consistent with the analytical comparison in Remark~\ref{rem:pareto_mrt_zf_comparison}.

\begin{figure}[t]
  \centering
  \includegraphics[width=0.85\linewidth]{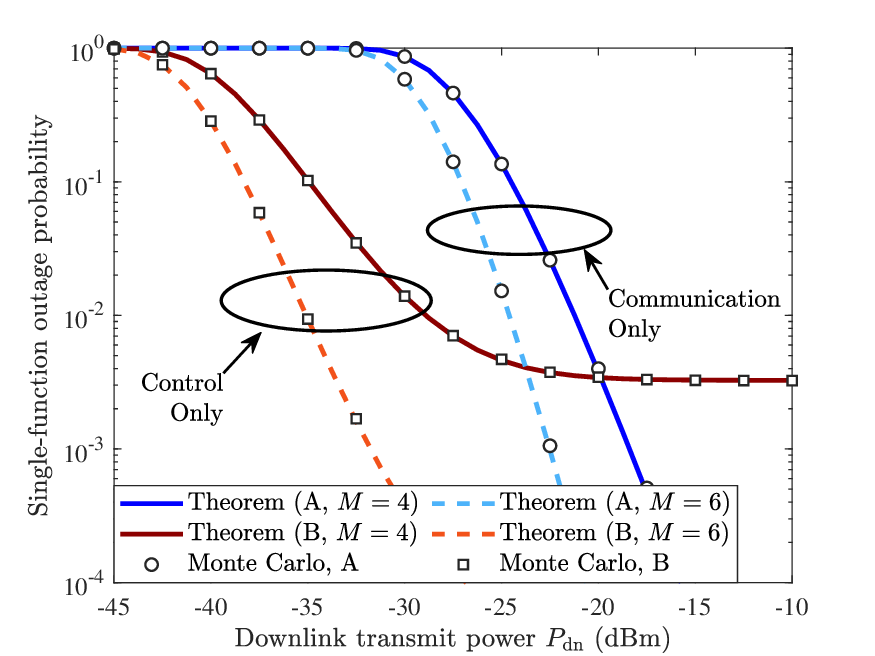}\vspace{-2mm}

  \caption{Outage probability for communication-only and control-only systems under $(\tau_{\mathrm{req}},V_{\mathrm{req}})=(10~\mathrm{ms},\,3\sigma_w^2)$.}
  \label{fig5_single_outage}
\end{figure}

Fig.~\ref{fig5_single_outage} illustrates the communication-only and control-only outage probabilities versus the downlink transmit power $P_{\mathrm{dn}}$ in the corresponding single-function systems. It can be observed that the Monte Carlo results closely match the analytical curves in all cases, thereby validating the theoretical results in \emph{Theorem~\ref{thm:ctrl_outage_mrt}}. For the communication-only function, the outage probability decreases markedly as $P_{\mathrm{dn}}$ increases, and this reduction becomes more pronounced when the number of BS antennas increases from $4$ to $6$. However, the same trend does not directly apply to the control-only function. In particular, when $M=4$, further increasing $P_{\mathrm{dn}}$ cannot continuously reduce the control outage probability, and a clear outage floor appears in the high-power region. The reason is that the control reliability is jointly constrained by both the uplink state-reporting quality and the downlink control transmission quality, so the uplink SNR becomes the dominant bottleneck. By contrast, when the number of BS antennas increases to $M=6$, the uplink reporting quality is also improved, and the control outage probability can therefore continue to decrease with $P_{\mathrm{dn}}$. These results reveal a key difference between communication reliability and control reliability: the latter is fundamentally limited by the joint uplink-downlink closed-loop quality.

\begin{figure}[t]
  \centering
  \includegraphics[width=0.85\linewidth]{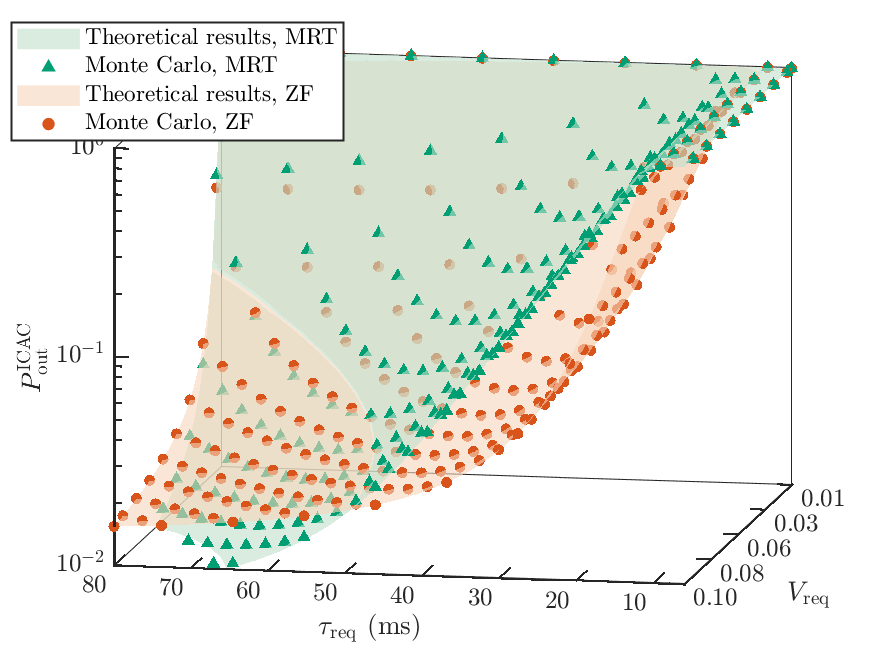}\vspace{-2mm}

  \caption{Joint JDCC outage probability surfaces versus the required communication delay and control error.}
  \label{fig6_fig7_joint_outage}
\end{figure}

Fig.~\ref{fig6_fig7_joint_outage} shows the theoretical joint JDCC outage probability surfaces under MRT and ZF, together with the corresponding Monte Carlo results. It is observed that the theoretical results match the simulations closely, thereby validating Theorems~\ref{thm:icac_mrt} and~\ref{thm:icac_zf}. Moreover, the joint outage probability increases rapidly as either the delay requirement $\tau_{\mathrm{req}}$ or the control-error requirement $V_{\mathrm{req}}$ becomes more stringent. This increase is more pronounced under MRT than under ZF.

\section{Conclusion}
\label{sec:conclusion}
This paper investigated the modeling and performance analysis of JDCC systems, where a multi-antenna BS simultaneously serves a CU and a CD within the same wireless architecture. Based on the developed JDCC model, the communication transmission delay and steady-state control variance were established as the two core performance metrics. More particularly, a rate-distortion-theoretic control analysis was developed to derive the variance recursion and the closed-form steady-state control variance. The Pareto boundary of the JDCC performance region was then characterized to reveal the fundamental communication-control trade-off, and closed-form expressions for JDCC performance regions under MRT and ZF were obtained. Moreover, the joint JDCC outage probability was defined to quantify the reliability of the JDCC system, and the corresponding outage expressions under MRT and ZF were derived. Numerical results validated the theoretical derivations and revealed that JDCC performance and reliability are jointly shaped by uplink reporting quality, downlink control and communication capabilities. The developed framework provides a unified analytical basis for understanding the performance trade-off and reliability coupling in JDCC systems.


\appendices
\newcommand{\appsection}[2]{%
	\refstepcounter{section}%
	\section*{APPENDIX~\thesection: #1}%
	\label{#2}%
}
\appsection{PROOF OF THEOREM~\ref{thm:bf_structure}}{app:proof_bf_structure}
For problem \eqref{eq:pareto_opt_sinr}, the KKT conditions are necessary at any regular Pareto-boundary point. Moreover, since $\Gamma_U^\star(\gamma_D)$ is obtained by maximizing the communication SINR under the control-SINR requirement $\gamma_D$, both the total-power constraint and the control-SINR constraint must be active at the optimum, which directly gives \eqref{eq:kkt_power} and \eqref{eq:kkt_sinr}. Next, setting the first-order derivative of the Lagrangian function \eqref{eq:lagrangian} with respect to $\mathbf{w}_D$ to zero yields $\mathbf{w}_D^\star \propto \left(\mathbf{I}_M + \frac{1}{\nu^\star}\frac{\Gamma_U^\star(\gamma_D)}{|\mathbf{h}_U^H\mathbf{w}_D^\star|^2+\sigma_{\mathrm{dn}}^2}\mathbf{h}_U\mathbf{h}_U^H\right)^{-1}\mathbf{h}_D$, which gives \eqref{eq:wD_opt} after defining $\mu_U$ as in \eqref{eq:kkt_muU} and introducing $p_D=\|\mathbf{w}_D^\star\|^2$. Similarly, setting the first-order derivative with respect to $\mathbf{w}_U$ to zero yields $\mathbf{w}_U^\star \propto \left(\mathbf{I}_M + \frac{1}{\nu^\star}\frac{\lambda^\star\gamma_D}{|\mathbf{h}_D^H\mathbf{w}_U^\star|^2+\sigma_{\mathrm{dn}}^2}\mathbf{h}_D\mathbf{h}_D^H\right)^{-1}\mathbf{h}_U$, which gives \eqref{eq:wU_opt} after defining $\mu_D$ as in \eqref{eq:kkt_muD} and setting $p_U=\|\mathbf{w}_U^\star\|^2$. Therefore, \eqref{eq:wD_opt}--\eqref{eq:kkt_muD} jointly characterize the beamforming solution that attains $\Gamma_U^\star(\gamma_D)$. Substituting the resulting $\Gamma_U^\star(\gamma_D)$ into the communication-delay expression and combining it with $V_\infty(\gamma_D)$ obtained from \eqref{eq:V_infinity} yields the Pareto-boundary point in \eqref{eq:pareto_point_gamma}.

\appsection{PROOF OF THEOREM~\ref{thm:mrt_pareto}}{app:proof_mrt_pareto}
Under the MRT scheme in \eqref{eq:bf_mrt}, the beam directions are fixed by the two channel vectors, and thus the resulting performance region is completely determined by the power allocation. Substituting \eqref{eq:bf_mrt} into the received-signal expressions yields
\begin{align}
	|\mathbf{h}_D^H \mathbf{w}_D^{\mathrm{MRT}}|^2 & = p_D^{\mathrm{MRT}} \|\mathbf{h}_D\|^2, \
	|\mathbf{h}_D^H \mathbf{w}_U^{\mathrm{MRT}}|^2 = p_U^{\mathrm{MRT}} \rho\,\|\mathbf{h}_D\|^2, \notag \\
	|\mathbf{h}_U^H \mathbf{w}_U^{\mathrm{MRT}}|^2 & = p_U^{\mathrm{MRT}} \|\mathbf{h}_U\|^2, \
	|\mathbf{h}_U^H \mathbf{w}_D^{\mathrm{MRT}}|^2 = p_D^{\mathrm{MRT}} \rho\,\|\mathbf{h}_U\|^2.
\end{align}
Under the sum-power constraint $p_D^{\mathrm{MRT}} + p_U^{\mathrm{MRT}} = P_{\mathrm{dn}}$, the communication SINR can be written as
\begin{equation}
	\Gamma_U(p_D)
	=
	\frac{(P_{\mathrm{dn}}-p_D)\|\mathbf{h}_U\|^2}
	{p_D\rho\|\mathbf{h}_U\|^2 + \sigma_{\mathrm{dn}}^2},
\end{equation}
where $p_D \in [0,P_{\mathrm{dn}}]$. It is straightforward to verify that $\Gamma_U(p_D)$ is strictly decreasing in $p_D$, whereas $\mathrm{SINR}_D^{\mathrm{dn}}(p_D)$ is strictly increasing in $p_D$. Therefore, for a given control-SINR target $\gamma_D$, minimizing the communication delay is equivalent to maximizing $\Gamma_U(p_D)$, which is achieved by selecting the smallest feasible power allocated to the control beam. As a result, the optimal solution must satisfy the control-SINR constraint with equality, i.e.,
\begin{equation}
	\frac{p_D^{\mathrm{MRT}}\|\mathbf{h}_D\|^2}
	{(P_{\mathrm{dn}}-p_D^{\mathrm{MRT}})\rho\|\mathbf{h}_D\|^2 + \sigma_{\mathrm{dn}}^2}
	=
	\gamma_D.
	\label{eq:sinr_D_mrt}
\end{equation}
Solving \eqref{eq:sinr_D_mrt} yields \eqref{eq:pD_mrt}, and the corresponding $p_U^{\mathrm{MRT}}$ follows from $p_U^{\mathrm{MRT}} = P_{\mathrm{dn}} - p_D^{\mathrm{MRT}}$. Substituting the resulting power allocation into the communication-delay expression and the steady-state variance expression then gives the stated MRT performance region.

\appsection{PROOF OF THEOREM~\ref{thm:zf_pareto}}{app:proof_zf_pareto}
Under the ZF scheme in \eqref{eq:zf_directions}, the beam directions are fixed by the two null-space projections, and thus the resulting performance region is completely determined by the power allocation. Substituting \eqref{eq:zf_directions} into the received-signal expressions yields
\begin{equation}
	|\mathbf{h}_D^H \mathbf{w}_D^{\mathrm{ZF}}|^2 \!\!=\! p_D^{\mathrm{ZF}} \|\mathbf{h}_D\|^2(1-\rho),
	\,
	|\mathbf{h}_U^H \mathbf{w}_U^{\mathrm{ZF}}|^2\!\! =\! p_U^{\mathrm{ZF}} \|\mathbf{h}_U\|^2(1-\rho),
	\label{eq:zf_gains}
\end{equation}
while the cross-user interference terms are eliminated by construction. Under the sum-power constraint $p_D^{\mathrm{ZF}} + p_U^{\mathrm{ZF}} = P_{\mathrm{dn}}$, the communication SINR becomes
\begin{equation}
	\Gamma_U(p_D)
	=
	\frac{(P_{\mathrm{dn}}-p_D)\|\mathbf{h}_U\|^2(1-\rho)}{\sigma_{\mathrm{dn}}^2},
\end{equation}
which is strictly decreasing in $p_D$ for $p_D \in [0,P_{\mathrm{dn}}]$, whereas the control SINR under ZF is strictly increasing in $p_D$. Therefore, for a given control-SINR target $\gamma_D$, minimizing the communication delay is equivalent to maximizing $\Gamma_U(p_D)$, which is achieved by selecting the smallest feasible power allocated to the control beam. As a result, the optimal solution must satisfy the control-SINR constraint with equality, i.e.,
\begin{equation}
	\frac{p_D^{\mathrm{ZF}}\|\mathbf{h}_D\|^2(1-\rho)}{\sigma_{\mathrm{dn}}^2}
	=
	\gamma_D,
	\label{eq:sinr_D_zf}
\end{equation}
which yields $p_D^{\mathrm{ZF}}=\frac{\gamma_D\,\sigma_{\mathrm{dn}}^2}{\|\mathbf{h}_D\|^2(1-\rho)}$ and thus \eqref{eq:pU_zf}. Substituting the resulting power allocation into the communication-delay expression and the steady-state variance expression then gives the stated ZF performance region.

\appsection{PROOF OF THEOREM~\ref{thm:icac_mrt}}{app:proof_joint_mrt_outage}
Conditioning on $(G_D,G_U,\rho)=(x,y,r)$ with $x>\eta_V$, and letting $p_U=P_{\mathrm{dn}}-p_D$, the two MRT downlink SINRs become
\begin{align}
	\mathrm{SINR}_D^{\mathrm{dn}}(p_D)
	&=
	\frac{p_D\beta_D x}{(P_{\mathrm{dn}}-p_D)\beta_D r x + \sigma_{\mathrm{dn}}^2},
	\label{eq:sinr_d_mrt_proof} \\
	\mathrm{SINR}_U^{\mathrm{dn}}(p_D)
	&=
	\frac{(P_{\mathrm{dn}}-p_D)\beta_U y}{p_D\beta_U r y + \sigma_{\mathrm{dn}}^2}.
	\label{eq:sinr_u_mrt_proof}
\end{align}
For fixed $(x,y,r)$, $\mathrm{SINR}_D^{\mathrm{dn}}(p_D)$ is strictly increasing in $p_D$, whereas $\mathrm{SINR}_U^{\mathrm{dn}}(p_D)$ is strictly decreasing in $p_D$. Therefore, the joint success event is equivalent to the existence of a power split $p_D \in [0,P_{\mathrm{dn}}]$ such that
\begin{equation}
	\mathrm{SINR}_D^{\mathrm{dn}}(p_D) \ge \gamma_D^{\mathrm{req}}(x),
	\qquad
	\mathrm{SINR}_U^{\mathrm{dn}}(p_D) \ge \gamma_U^{\mathrm{req}},
\end{equation}
or equivalently, the feasible interval for $p_D$ is nonempty.

From \eqref{eq:sinr_d_mrt_proof}, the control constraint gives the lower bound
\begin{equation}
	p_D \ge p_D^{\min}(x,r)
	\triangleq
	\frac{\gamma_D^{\mathrm{req}}(x)\bigl(P_{\mathrm{dn}}\beta_D r x + \sigma_{\mathrm{dn}}^2\bigr)}
	{\beta_D x\bigl(1+\gamma_D^{\mathrm{req}}(x)r\bigr)},
	\label{eq:pD_min_mrt_joint}
\end{equation}
while \eqref{eq:sinr_u_mrt_proof} gives the upper bound
\begin{equation}
	p_D \le p_D^{\max}(y,r)
	\triangleq
	\frac{P_{\mathrm{dn}}\beta_U y - \gamma_U^{\mathrm{req}}\sigma_{\mathrm{dn}}^2}
	{\beta_U y(1+\gamma_U^{\mathrm{req}}r)}.
	\label{eq:pD_max_mrt_joint}
\end{equation}
Hence, joint success holds if and only if
\begin{equation}
	p_D^{\min}(x,r) \le p_D^{\max}(y,r).
	\label{eq:mrt_feas_interval}
\end{equation}

Substituting \eqref{eq:pD_min_mrt_joint} and \eqref{eq:pD_max_mrt_joint} into \eqref{eq:mrt_feas_interval}, and rearranging, we obtain
\begin{equation}
	\begin{aligned}
		 y\!\left[x\bigl(1-\gamma_U^{\mathrm{req}}\gamma_D^{\mathrm{req}}(x)r^2\bigr)-\eta_f(x)\bigl(1+\gamma_U^{\mathrm{req}}r\bigr)\right]
		\ge \\
		\eta_U x\bigl(1+\gamma_D^{\mathrm{req}}(x)r\bigr).
	\end{aligned}
	\label{eq:mrt_feas_rearranged}
\end{equation}
Since the right-hand side is strictly positive, a feasible $y$ exists only if
\begin{equation}
	1-\gamma_U^{\mathrm{req}}\gamma_D^{\mathrm{req}}(x)r^2 > 0,
	\qquad
	x > \frac{\eta_f(x)\bigl(1+\gamma_U^{\mathrm{req}}r\bigr)}{1-\gamma_U^{\mathrm{req}}\gamma_D^{\mathrm{req}}(x)r^2}
	= \xi(x,r).
\end{equation}
These conditions are equivalent to $0 \le r < r_{\max}(x)$ and $x>\xi(x,r)$. Under them, \eqref{eq:mrt_feas_rearranged} reduces to $y \ge \psi(x,r)$, where $\psi(x,r)$ is given in \eqref{eq:psi_mrt}. Therefore, conditioning on $(G_D,\rho)=(x,r)$, the joint success event is equivalent to
\begin{equation}
	\left\{0 \le r < r_{\max}(x),\ x>\xi(x,r),\ G_U \ge \psi(x,r)\right\}.
\end{equation}
The corresponding conditional success probability is
\begin{equation}
	\mathbf{1}\!\left(x>\xi(x,r)\right)\bigl[1-F_G\!\bigl(\psi(x,r)\bigr)\bigr]
\end{equation}
for $0 \le r < r_{\max}(x)$, and zero otherwise. Integrating over $r \in [0,r_{\max}(x)]$ and $x \in (\eta_V,\infty)$ yields \eqref{eq:icac_mrt_out}.

\appsection{PROOF OF THEOREM~\ref{thm:icac_zf}}{app:proof_joint_zf_outage}
Conditioning on $(G_D,G_U,\rho)=(x,y,r)$ with $x>\eta_V$ and letting $p_U=P_{\mathrm{dn}}-p_D$, the ZF downlink SINRs are
\begin{align}
	\mathrm{SINR}_D^{\mathrm{dn}}(p_D)
	&=
	\frac{p_D\beta_D(1-r)x}{\sigma_{\mathrm{dn}}^2},
	\label{eq:sinr_d_zf_proof} \\
	\mathrm{SINR}_U^{\mathrm{dn}}(p_D)
	&=
	\frac{(P_{\mathrm{dn}}-p_D)\beta_U(1-r)y}{\sigma_{\mathrm{dn}}^2}.
	\label{eq:sinr_u_zf_proof}
\end{align}
For fixed $(x,y,r)$, $\mathrm{SINR}_D^{\mathrm{dn}}(p_D)$ is strictly increasing in $p_D$, whereas $\mathrm{SINR}_U^{\mathrm{dn}}(p_D)$ is strictly decreasing in $p_D$. Therefore, joint success is equivalent to the existence of a power split $p_D \in [0,P_{\mathrm{dn}}]$ such that
\begin{equation}
	\mathrm{SINR}_D^{\mathrm{dn}}(p_D) \ge \gamma_D^{\mathrm{req}}(x),
	\qquad
	\mathrm{SINR}_U^{\mathrm{dn}}(p_D) \ge \gamma_U^{\mathrm{req}},
\end{equation}
or equivalently, the feasible interval for $p_D$ is nonempty.

From \eqref{eq:sinr_d_zf_proof} and \eqref{eq:sinr_u_zf_proof}, the two constraints give
\begin{equation}
	p_D \ge \frac{\gamma_D^{\mathrm{req}}(x)\sigma_{\mathrm{dn}}^2}{\beta_D(1-r)x},
	\qquad
	p_D \le P_{\mathrm{dn}} - \frac{\gamma_U^{\mathrm{req}}\sigma_{\mathrm{dn}}^2}{\beta_U(1-r)y}.
\end{equation}
Hence, joint success holds if and only if
\begin{equation}
	y\bigl[(1-r)x-\eta_f(x)\bigr] \ge \eta_U x.
	\label{eq:zf_feas_rearranged}
\end{equation}
Since the right-hand side is strictly positive, a finite $y$ can satisfy \eqref{eq:zf_feas_rearranged} only if $(1-r)x>\eta_f(x)$, namely
\begin{equation}
	0 \le r < 1-\frac{\eta_f(x)}{x}.
\end{equation}
Under this condition, \eqref{eq:zf_feas_rearranged} reduces to $y \ge \phi(x,r)$. Therefore, conditioning on $(G_D,\rho)=(x,r)$, the joint success event is equivalent to
\begin{equation}
	\left\{0 \le r < 1-\frac{\eta_f(x)}{x},\ G_U \ge \phi(x,r)\right\}.
\end{equation}
Using the independence of $G_D$, $G_U$, and $\rho$, the corresponding success probability is $1-F_G(\phi(x,r))$ for $0 \le r < 1-\eta_f(x)/x$, and zero otherwise. Integrating over $r \in [0,[1-\eta_f(x)/x]^+]$ and $x \in (\eta_V,\infty)$ yields \eqref{eq:icac_zf_out}.

\bibliographystyle{IEEEtran}
\bibliography{reference/mybib}

\end{document}